\documentclass[a4paper,12pt]{amsart}
\usepackage[utf8]{inputenc}
\usepackage[all]{xy}
\usepackage{amssymb}
\usepackage{amsmath}
\usepackage{csquotes}
\usepackage{subfig}
\usepackage{enumerate}
\usepackage{graphicx}
\usepackage{enumerate}
\usepackage{amsaddr}
\usepackage{color}
\newtheorem{lemma}{Lemma}

\newtheorem{theorem}{Theorem}

\usepackage{cancel}
\usepackage{soul}

\title[Additive hazard and MSMs in continuous time]{The additive hazard estimator is consistent for continuous-time marginal structural models}

\author{Pål C. Ryalen, Mats J. Stensrud, and Kjetil Røysland}
\address{Department of Biostatistics, University of Oslo, Domus Medica Gaustad, Sognsvannsveien 9, 0372 Oslo, Norway}
\date{\today}

\MakeOuterQuote{"}\EnableQuotes
\DeclareUnicodeCharacter{00A0}{ }
\DeclareMathOperator{\Tr}{Tr}

\begin{document}


\begin{abstract}
Marginal structural models (MSMs) allow for causal analysis of longitudinal data. The MSMs were originally developed as discrete time models. Recently, continuous-time MSMs were presented as a conceptually appealing alternative for survival analysis. In applied analyses, it is often assumed that the theoretical treatment weights are known, but these weights are usually unknown and must be estimated from the data. Here we provide a sufficient condition for a class of continuous-time MSMs to be consistent even when the weights are estimated, and we show how additive hazard models can be used to estimate such weights. Our results suggest that the continuous-time weights perform better than IPTW when the underlying treatment process is continuous. Furthermore, we may wish to transform effect estimates of hazards to other scales that are easier to interpret causally. We show that a general transformation strategy can be used on weighted cumulative hazard estimates to obtain a range of other parameters in survival analysis, and demonstrate how this strategy can be applied on data using our \texttt{R} packages \texttt{ahw} and \texttt{transform.hazards}.
\end{abstract}

\maketitle

\section{Outline}

 

MSMs can be used to obtain causal effect estimates in the presence of confounders, which e.g.\ may be time-dependent \cite{robins2000marginal}. The procedure is particularly appealing because it allows for a sharp distinction between confounder adjustment and model selection \cite{joffe2004model}: first, we adjust for observed confounders by weighing the observed data to obtain balanced pseudopopulations. Then, we calculate effect estimates from these pseudopopulations based on our structural model. 

Traditional MSM techniques for survival analysis have considered time to be a discrete processes \cite{hernan2000marginal}. In particular, inverse probability of treatment weights (IPTWs) are used to create the pseudopopulations, and then e.g.\ several subsequent logistic regressions are fitted for discrete time intervals to mimic a proportional hazards model.

However, time is naturally perceived as a continuous process, and it also seems natural to analyse time-to-event outcomes with continuous models. Inspired by the discrete time MSMs, R\o ysland suggested a continuous-time analogue to MSMs \cite{roysland2011}. Similar to the discrete MSMs, it has been shown that the continuous MSMs can be used to obtain consistent effect estimates when the theoretical treatment weights are known \cite{roysland2011}. In particular, additive hazard regressions can be weighted with the theoretical continuous-time weights to yield consistent effect estimates. Nevertheless, the weights are usually unknown in real life and must be estimated from the data. To the best of our knowledge, the performance of MSM when the IPTW are estimated remains to be elucidated. 

In this article, we show that continuous-time MSMs also perform desirable when the treatment weights are estimated from the data: we provide a sufficient condition to ensure that weighted additive hazard regressions are consistent. Furthermore, we show how such weighted hazard estimates can be consistently transformed to obtain other parameters that are easier to interpret causally. To do this, we use stability theory of SDEs, which allows us to target a range of parameters expressed as solutions of ordinary differential equations. Many examples of such parameters can be found in \cite{ryalen2018transforming}. This is immediately appealing for causal survival analysis: First, we can use hazard models, that are convenient for regression modeling, to obtain weights. Estimates on the hazard scale are hard to interpret causally per se \cite{robins1989probability,hernan2010hazards,aalen2015does,stensrud2017exploring}, but we present a generic method to consistently transform these effect estimates to several other scales that are easier to interpret.

The continuous-time weights and the causal parameters can be estimated using the \texttt{R} package \texttt{ahw}. We show that this \texttt{ahw} weight estimator, which is based on additive hazard regression, is consistent in Theorem \ref{thm:ahwConsist}. We have implemented code for transforming cumulative hazard estimates in the package \texttt{transform.hazards}. These packages make continuous-time marginal structural modeling easier to implement for applied researchers.

\section{Weighted additive hazard regression}
\label{section: Derivations of additive hazards}

\subsection{Motivation}

We will present a strategy for dealing with confounding and dependent censoring in continuous time. Confounding, which may be time-varying, will often be a problem when analysing observational data, e.g.\ coming from health registries. The underlying goal is to assess the effect a treatment strategy has on an outcome.

We can describe processes in continuous time using local (in)dependence relations, and we can use local independence graphs to visualise these relations. A precise description of local independence can be found in \cite{roysland2011}. The local independence graph we will focus on is
\begin{equation*}\label{gr:local independence graph}
               \xymatrix{
           C 
           & \mathcal L \ar@/_/[l] \ar@/^/[d] \ar@/^/[r] & D \\
           & A \ar@/^/[ul] \ar@/^/[u]  \ar@/_/[ur] & }.
\end{equation*} Heuristically, the time-dependent confounders $\mathcal L$ and the exposure $A$ can influence the censoring process $C$ and the event of interest $D$. Moreover, the time-dependent confounders can both influence and be influenced by the exposure process. We include baseline variables, some of which may be confounders, in Section \ref{section:Hypothetical scenarios and likelihood ratios}. 

The above graph can e.g.\ describe a follow-up study of HIV-infected subjects, where the initiation and adjustment of HIV treatment depend on CD4 count measurements over time \cite{hernan2000hiv}. The CD4 count is a predictor of future survival, and it is also a diagnostic factor that informs initiation of zidovudine treatment; a CD4 count below a certain threshold indicates that treatment is needed. The CD4 count will, in turn, tend to increase in response to treatment, and is monitored over time to inform the future treatment strategy. Hence, it is a time-dependent confounder. In most follow-up studies there is a possibility for subjects to be censored, and we allow the censoring to depend on the covariate and treatment history, as long as subjects are alive.

In \cite{ryalen2018pcancer} we analysed a cohort of Norwegian males diagnosed with prostate cancer, using the theory from this article to compare treatment effectiveness of radiation and surgery, even though time-dependent confounding were thought to be a minor issue. The continuous-time MSMs allowed us to estimate causal cumulative incidences on the desired time-scale, starting from the time of diagnosis. This example shows that (continuous-time) MSMs can also be a preferable choice in the absence of time-dependent confounding.

\subsection{Hypothetical scenarios and
        likelihood ratios}
\label{section:Hypothetical scenarios and likelihood ratios}

We consider observational event-history data where $n$ i.i.d. subjects are followed over the study period $[0,T]$. Let $N^{i,A} \text{ and } N^{i,D}$ respectively be counting processes that jump when treatment $A$ and outcome $D$ of interest occur for subject $i$. Furthermore, let $Y^{i,A},Y^{i,D}$ be the at-risk processes for $A$ and $D$. We let $\mathcal{V}_0$ be the collection of baseline variables that are not confounders, as well as the treatment and outcome processes. $\mathcal{L}$ are the (time-dependent) confounders. For now, we assume independent censoring, but we will show how our methods can be applied in some scenarios with dependent censoring in Section \ref{sec:censoring weights}.

Let $\mathcal F_t^{i, \mathcal V_0 \cup \mathcal L}$ denote the filtration that is generated by all the observable events for individual $i$. Moreover, let $P^i$ denote the probability measure on $\mathcal F_T ^{i, \mathcal V_0 \cup \mathcal L}$ that governs the frequency of observations of these events, and $\lambda_t^{i,D}$ denote the intensity for $N^{i,D}$ with respect to $P^i$ and the filtration $\mathcal F_t ^{i, \mathcal V_0 \cup \mathcal L}$.

We aim to estimate the outcome in a hypothetical situation where a treatment intervention is made according to a specified strategy. Suppose that the frequency of observations we would have seen in this hypothetical scenario is described by another probability measure $\tilde P^i$ on $\mathcal F_T ^{i, \mathcal V_0 \cup \mathcal L}$. Furthermore, assume that all the individuals are also i.i.d.\ in the hypothetical scenario and that $\tilde P^i \ll P^i$, i.e.\ that there exists a likelihood ratio 
$$R^i_t := \frac {d\tilde
        P^i| _{ \mathcal F_t ^{i, \mathcal V_0 \cup \mathcal L}  }  }{d P^i| _{
                \mathcal F_t ^{i, \mathcal V_0 \cup \mathcal L}  } }$$ for each time $t$. 
We will later describe how an explicit form of $\{R^i\}_i$ can be obtained. It relies on the assumption that the underlying model is causal, a concept we define in Section \ref{section:causal validity}. For the moment we will not require this, but only assume that $\lambda_t^{i,D}$ defines the intensity with respect to $\mathcal F_t ^{i, \mathcal V_0 \cup \mathcal L}$ for both $P^i$ and $\tilde P^i$; that is, the functional form of $\lambda_t^{i,D}$ is identical under both $P^i$ and $\tilde P^i$. 

Suppose that $N^{i,D}$ has an additive hazard with respect to $\tilde P^i$ and
the filtration $\mathcal F_t ^{i, \mathcal V_0}$ that is generated by the {the components of}  $\mathcal V_0$. We stress that we consider the intensity process marginalized over $\mathcal L$, and thereby it is defined with respect to $\mathcal F_t ^{i, \mathcal V_0}$, and not $\mathcal F_t ^{i, \mathcal V_0 \cup \mathcal L}$.   In other words, we assume that the hazard for event $D$ with respect to the filtration $\mathcal F_t ^{i, \mathcal V_0}$ is additive, and can be written as
\begin{align}
	X^{i \intercal}_{t-} b_t, \label{eq:additiveOutcomeHazard}
\end{align}
where $b_t$ is a bounded and continuous vector valued function, and the components of $X^{i}$ are covariate processes or baseline-variables from $\mathcal V_0$.

\subsection{Re-weighted additive hazard regression}
\label{section:reweighting}
Our main goal is to estimate the cumulative coefficient function in \eqref{eq:additiveOutcomeHazard}, i.e.\
\begin{align}
	B_t := \int_0^t b_s ds \label{eq:cumulativeCoefficient}
\end{align}
from the observational data distributed according to $ P = P^1 \otimes \dots \otimes P^n$. If we had known all the true likelihood ratios, we could try to estimate \eqref{eq:cumulativeCoefficient} by re-weighting each individual in Aalen's additive hazard regression \cite[VII.4]{Andersen} according to its likelihood ratio. However, the true weights are unlikely to be known, even if the model is causal. In real-life situations, we can only hope to have consistent estimators for these weights. We therefore consider $ \mathcal F_{t}^{1,\mathcal  V_0 \cup \mathcal L  } \otimes \cdots \otimes \mathcal F_{t}^{n,\mathcal V_0 \cup \mathcal L }$-adapted estimates  $\{ R^{(i,n)}_t \} _{n}$ that converge to $R^i_t$ under relatively weak assumptions, such that Aalen's additive hazard regression for the outcome re-weighted according to $\{R^{(i,n)}_t\}$ gives consistent estimates of the causal cumulative hazard. The estimator we will consider is defined as follows: Let $N^{(n)}$ be the vector of counting processes and $X^{(n)}$ the matrix containing the $X^i$'s, that is,
\begin{equation}
  N^{(n)}_t := \begin{pmatrix}
  N^{1,D}_t \\ \vdots \\N^{n,D}_t
      \end{pmatrix} \text{ and }
X_s^{(n)} := \begin{pmatrix}
        X^{1,1}_s &  \dots & X^{1,p}_s \\
        \vdots&  & \vdots \\ 
	    X^{n,1}_s &  \dots & X^{n,p}_s
\end{pmatrix},
\end{equation}
and let  $Y^{(n),D}_s$ denote the $n\times n$-dimensional diagonal matrix where the $i$'th diagonal element is $Y^{i,D}_s \cdot R^{(i,n)}_{s-}$. The weighted additive hazard regression is given by: 
\begin{equation}
	\label{eq:weighted cumulative hazard estimator} B_t ^{(n)} := \int_0^t
        (X_{s-}^{(n)\intercal }  Y^{(n),D}_{s} X_{s-}^{(n)})^{-1}
        X_{s-}^{(n)\intercal}   Y^{(n),D}_{s} d N_s^{(n)}. 
\end{equation}




\subsubsection{Parameters that are transformations of cumulative hazards}
\label{section:transforming hazards}
It has recently been emphasised that the common interpretation of hazards in survival analysis as the causal risk of death during $(t, t + \Delta]$ for an individual that is alive at $t$, is often not appropriate; see e.g.\  \cite{hernan2010hazards}. An example in \cite{aalen2015does} shows that this can also be a problem in RCTs; if $N$ is a counting process that jumps at the time of the event of interest, $A$ is a randomised treatment, and $U$ is an unobserved frailty, the following causal diagram describes such a situation:
\begin{equation*}\label{gr:collider}
               \xymatrix{
           A \ar[r] \ar@/_/[dr] & N_t \ar[d] & \ar[l]  \ar@/^/[dl] U \\
           & N_{t+ \Delta}& }.
\end{equation*}
If we consider the probability of an event before $N_{t + \Delta}$, conditioning on no event at time $t$, we condition on a collider that opens a non-causal path from $A$ to the outcome. This could potentially have dramatic consequences since much of survival analysis is based on the causal interpretation of hazards, e.g.\ hazard ratios.

In \cite{ryalen2018transforming}, we have suggested a strategy to handle this situation: even if it is difficult to interpret hazard estimates causally per se, we can use hazard models to obtain other parameters that have more straightforward interpretations. Population based measures such as the survival function, the cumulative incidence functions, and the restrictive mean survival function, do not condition on survival and will therefore not be subject to the selection bias. Moreover, these measures, and many others (see \cite{ryalen2018transforming,stensrud2018nullhypothesis} for examples), solve differential equations driven by cumulative hazards, i.e.\ they are functions $\eta_t$ that can be written on the form
\begin{equation}  \label{eq:ODE}
        {\eta}_t = {\eta}_0 + \int_0^t F({\eta}_{s-}) d{B}_{s}, 
\end{equation}
where $B$ are cumulative hazard coefficients, and $F$ is a Lipschitz continuous matrix-valued function. In \cite{ryalen2018transforming}, we showed how to estimate $\eta$ by replacing the integrator in \eqref{eq:ODE} with an estimator $B^{(n)}$ that can be written as a counting process integral. Examples of such $B^{(n)}$ include the Nelson-Aalen, or more generally Aalen's additive hazard estimator. This gave rise to the stochastic differential equation
\begin{equation}  \label{eq:SDE}
        \eta_t^{(n)} = \eta_0^{(n)} + \int_0^t F(\eta_{s-}^{(n)})
        dB_s^{(n)}, 
\end{equation}
that is easy to solve on a computer; it is a piecewise constant, recursive equation that jumps whenever the integrator $B^{(n)}$ jumps. Hence, \eqref{eq:SDE} can be solved using a \texttt{for} loop over the jump times of $B^{(n)}$, i.e.\ the survival times of the population.

A simple example of a parameter on the form \eqref{eq:ODE} is the survival function, which reads $S_t = 1 - \int_0^t S_{s} dB_s$, where $B$ is the cumulative hazard for death. In this case, the estimation strategy \eqref{eq:SDE} yields the Kaplan-Meier estimator. Nevertheless, some commonly studied parameters cannot be written on the form \eqref{eq:ODE}, such as the median survival, and the hazard ratio.

In \cite{ryalen2018transforming} we showed that $\eta^{(n)}$ provides a consistent estimator of $\eta$ if
\begin{itemize} 
\item $\lim_{ n \rightarrow \infty} P (    \sup_{t\leq T} | B_t^{(n)} -  B_t | \geq \epsilon
) = 0 $ for every $\epsilon >0$, i.e.\ the cumulative hazard estimator is consistent, and
\item the estimator $ B^{(n)}$ is
predictably uniformly tight, abbreviated P-UT. 
\end{itemize}
The additive hazard estimator satisfies both these criteria, and additive hazard regression can thus be used as an intermediate step for flexible estimation of several parameters, such as the survival, the restricted mean survival, and the cumulative incidence functions \cite{ryalen2018transforming}. In Theorem \ref{thm:AalenConsistency}, we show that also the re-weighted additive hazard regression satisfies these properties, which is a major result in this article. Thus, we can calculate causal cumulative hazard coefficients, and transform them to estimate MSMs that solve ordinary differential equations consistently. In Section \ref{section:a marginal structural model} we illustrate how such estimation can be done, by including an example of a marginal structural relative survival model on simulated data.

A mathematically precise definition of P-UT is given in \cite[VI.6a]{JacodShiryaev}. 
We will not need the full generality of this definition here. Rather, we will use
\cite[Lemma 1]{ryalen2018transforming} to determine if processes are P-UT. The Lemma states that whenever $\{ J_t^{(n)}\}_n$ is  a sequence of semi-martingales  on
              $[0,T]$ with Doob-Meyer decompositions  
  $$ J^{(n)}_t = \int_0 ^t  \rho_s^{(n)}ds +  M^{(n)}_t ,$$
  where $\{ M^{(n)}\}_n$ are  square integrable local martingales and 
  $\{ \rho^{(n)}\}_n$ are  predictable processes such that 
                                                  \begin{equation}
                                                          \label{eq:tightRho}
                                      \lim_{a \rightarrow \infty} \sup_n P \bigg( \sup_s  | \rho_s^{(n)} |_1 \geq a   \bigg) = 0
                                      \text{ and}  
                      \end{equation}

                        \begin{equation} \label{eq:PUTMart}
                                      \lim_{a \rightarrow \infty} \sup_n P
                              \bigg(   \Tr 
                              \langle  M^{(n)} \rangle_T  \geq a   \bigg) = 0, 
                      \end{equation}
                        then  $\{ J_t^{(n)}\}_n$ is P-UT. Here, $\Tr$ is the trace function, and $\langle \cdot \rangle$ is the predictable variation.

 \subsection{Consistency and  P-UT property}

The consistency and P-UT property of $B^{(n)}$ introduced in Section \ref{section:reweighting} is stated as a Theorem below. A proof can be found in the Appendix.

 \begin{theorem}[Consistency of weighted additive hazard regression] \label{thm:AalenConsistency}
                       Suppose that 
                       \begin{enumerate}[I)]
                               \item   \label{enum:equalFooting}
                                       The conditional density of $R^{(i,n)}_t$
                                       given   ${ \mathcal{F}_t^{i, \mathcal{V}_0 \cup \mathcal{L} } }$ 
                         does not depend on $i$, 
 
                                       \item  \label{eq:lambdabound} 
                                   \begin{equation*}
                                            E_P [ \sup_{t \leq T} | \lambda_t^{ {1},D} |^2  ]  <
                                      \infty ~\text{and } E_P [ \sup_{t \leq T} | X_t^{1} |^2  ]  <
                                      \infty 
                                   \end{equation*}
                              \item \label{eq:boundedmatrix} Let 
                                      \begin{equation*}
          \Gamma^{(n)}_t := \bigg(\frac 1 n   X^{(n)\intercal}_{t-} Y^{(n),D}_t
          X^{(n)}_{t-} \bigg) = 
                      \begin{pmatrix}
                              \frac{1}{n} \sum_{k = 1}^n 
                              R^{(k,n)}_{t-}X_{t-}^{k,i} Y^{k,D}_t X_{t-}^{k,j}
                      \end{pmatrix}_{i,j}, 
              \end{equation*}
and suppose that \begin{equation*} \lim_{ a \rightarrow \infty} \inf_n
                                      P\bigg( \sup_{t \leq T}  
                                              \Tr \big( \Gamma_t^{(n) -1}
                                              \big) 
                                              > a \bigg) =  0, 
                                      \end{equation*}
                               \item \label{eq:consistW}
                                       Suppose that $\{R^i\}_i$ and
                                      $\{R^{(i,n)}\}_{i,n}$ are
                                      uniformly bounded and 
                                      \begin{equation}
                                            \lim\limits_{n\longrightarrow \infty} P \big( \big|R^{(i,n)}_t -
                                            R^i_t \big| >
                                            \delta
                                            \big) = 0
                                        \end{equation}
                                        for every $i$, $\delta > 0$ and $t$. 
                                            
                       \end{enumerate}

                       Then $\{{B}^{(n)}\}_n$ is P-UT and 
                       \begin{equation}
                               \lim\limits_{n\longrightarrow \infty} P \bigg( \sup_{t \leq T}   \big| B^{(n)}_t - B_t    \big| \geq \delta
                               \bigg)  = 0,
                       \end{equation}
for every $\delta > 0$.

      \end{theorem}
      
Heuristically, condition $\ref{enum:equalFooting})$ states that if we know individual $i$'s realisation of the variables and processes in $\mathcal{V}_0 \cup \mathcal{L}$ up to time $t$, no other information on individual $i$ is used for estimating her weight at $t$. Condition $\ref{eq:lambdabound})$ ensures that the number of outcome events will not blow up, or suddenly grow by an extreme amount. Condition $\ref{eq:boundedmatrix})$ implies that there can be no collinearity among the covariates, or more precisely that the inverse matrix of $\big(E[ X^{1,i}_t X^{1,j}_t]\big)_{i,j}$ is uniformly bounded in $t$. Condition $\ref{eq:consistW})$ states that the weight estimator converges to the theoretical weights $R_t^i$, in a not very strong sense. The uniform boundedness of $\{ R^i \}_i$ is a positivity condition similar to the positivity condition required for standard inverse probability weighting.

\section{Causal validity and a consistent estimator the individual likelihood ratios}
        \label{section:causal validity}

We can model the individual likelihood ratio in many settings where the underlying model is causal. To do this, we assume that each subject is represented by the outcomes of $r$ baseline variables $Q_1, \dots Q_r$, and $d$ counting processes $N^1, \dots, N^d$. Moreover, we let $\mathcal F_t$ denote the filtration that is generated by all their possible events before $t$. 

Suppose that $\lambda^1, \dots, \lambda^d$ are the intensities of the counting processes $N^1, \dots, N^d$ with respect to the filtration $\mathcal F_t$ and the observational probability $P$. Now, by \cite{jacod1975}, $P|_{\mathcal F_T }$ is uniquely determined by all the intensities and the conditional densities at baseline of the form $dP\big( Q^{k} | Q^{k-1}, \dots, Q^{1} \big)$, because the joint density at baseline factorises as a product of conditional densities.

Suppose that the observational scenario, where the frequency of events are described by $P$, is subject to an intervention on the component represented by $N^j$. Our model is said to be \textbf{causal} if such an intervention would not change the 'local characteristics' of the remaining nodes. More precisely this means that 
\begin{itemize}
    \item The functional form of the intensities on which we do not intervene coincide under $P$ and the intervened scenario $\tilde P$, i.e.\ $\lambda^{k}$ would also define the intensity for $N^{k}$ with respect to $\tilde P$ when $k \neq j$, and 

    \item  The conditional density of each $Q^{k}$, given
    $Q^{k-1}, \dots, Q^{1}$ would be  the same with
    respect to both  $P$ and $\tilde P$, i.e.\  $$dP\big( Q^{k}
    | Q^{k-1}, \dots, Q^{1} \big) = 
    d\tilde  P\big( Q^{k}
    | Q^{k-1}, \dots, Q^{1} \big)$$
    for $k=1,\cdots r$.
\end{itemize}

If the intervention instead were targeted at a
  baseline variable, say $Q^{j}$, and this intervention
  would replace   $dP\big( Q^{k} | Q^{k-1}, \dots, Q^{1} \big)$
  by $ d\tilde P\big( Q^{k} | Q^{k-1}, \dots, Q^{1} \big)$, for $k=1,\cdots r$, the model is said to be causal if
   \begin{itemize}
        \item 
        The intensity process for $N^{k}$ with respect to $P$ and $\tilde P$ coincide for all $k=1,\cdots p$, and 
        \item The remaining conditional densities at baseline
        coincide, i.e.\  
        $$dP\big( Q^{k}
        | Q^{k-1}, \dots, Q^{1} \big) = 
        d\tilde  P\big( Q^{k}
        | Q^{k-1}, \dots, Q^{1} \big),
        $$  
        for $k \neq j$.

\end{itemize}
        Note that the latter is in agreement with Pearl's definition of a causal model \cite{pearl}.

This notion of causal validity leads to an explicit formula for the likelihood ratio. If the intervention is aimed at $N^j$, changing the intensity from $\lambda^j$ to $\tilde \lambda^j$, then the likelihood ratio takes the form
\begin{equation} \label{eq:JACOD}
        R_t = \big( \prod_{s \leq t} \theta_s^{\Delta N_s^j}\big) \exp\big( \int_0^t \lambda_s^j -
        \tilde \lambda_s^j ds \big),
\end{equation}
where $\theta_t := \frac{\tilde \lambda_t^j }{\lambda_t^j }$, see \cite{roysland2011} and \cite{jacod1975}. 

If the intervention is targeted at a baseline variable, the likelihood ratio corresponds to the ordinary propensity score
         \begin{equation}\label{eq:propensity score}
                R_0 := \frac {d\tilde P\big( Q^{j}
                        | Q^{j-1}, \dots, Q^{1} \big)} 
                {dP\big( Q^{j}
                        | Q^{j-1}, \dots, Q^{1} \big)}. 
\end{equation}

Interventions on several nodes yield a likelihood ratio that is a product of terms on the form \eqref{eq:JACOD} and \eqref{eq:propensity score}. The terms in the product could correspond to baseline interventions, time-dependent treatment interventions, or interventions on the censoring intensity. It is natural to estimate the likelihood ratio, or weight process by a product of baseline weights, treatment weights, and censoring weights.

We want, of course, to identify the likelihood ratio that corresponds to $\tilde P$, as this is our strategy to assess the desired randomised trial. Following equations \eqref{eq:JACOD} and \eqref{eq:propensity score}, we see that the intervened intensities and baseline variables must be modeled correctly, and specifically that a sufficient set of confounders must be included when modeling the treatment intensity. Additionally, the MSM for the outcome must be correctly specified. An important consequence of the results in this paper is that a class of MSM parameters that solve ODEs driven by cumulative hazards can be estimated consistently.


As long as the intervention acts on a counting process or a baseline variable, the same formula would hold in much more general situations where the remaining covariates are represented by quite general stochastic processes. The assumption of 'coinciding intensities' must then be replaced by the assumption that the 'characteristic triples', a generalisation of intensities to more general processes, coincides for $P$ and $\tilde P$; see \cite[II.2]{JacodShiryaev}. 

        \subsection{Estimation of weights using additive hazard regression}
\label{section:additive hazard weight estimator}

Suppose we have a causal model as described in the beginning of Section \ref{section:causal validity}, allowing us to obtain a known form of the likelihood ratio $R^i$. To model the hypothetical scenario, we need to rely on estimates of the likelihood ratio. In the following, we will only focus on a causal model where we replace the intensity of treatment by $\tilde{\lambda}^{i,A}$, the intensity of $N^{i,A}$ with respect to $P$ and the subfiltration $\mathcal F_t ^{\mathcal V_0}$. It is a consequence of the innovation theorem \cite{Andersen} that $E[\lambda^{i,A}_t | \mathcal F_{t-} ^{\mathcal V_0} ] = \tilde \lambda_t^{i,A}$. Moreover, an exercise in asymptotics of stochastic processes shows that if we discretise time, the associated marginal model structural weights from \cite{robins2000marginal} approximate \eqref{eq:JACOD} gradually as the time-resolution increases.

We will not follow the route of \cite{robins2000marginal} to estimate $R^i$. Instead, we will use that \eqref{eq:JACOD} is the unique solution to the stochastic differential equation
        \begin{align*}
	R_t^i  &= R_0^i + \int_0^t R_{s-}^i  dK_s^i  \\
    K_t^i  &= \int_0^t (\theta_s^i - 1) dN_s^{i,A} + 
    \int_0^t \lambda_s^{i,A} ds -  \int_0^t \tilde \lambda_s^{i,A} ds,
    \end{align*}

with $\theta^i = \frac{\tilde{\lambda}^{i,A}}{\lambda^{i,A}}$. To proceed, we assume that $\lambda^{i,A}$ and $\tilde
        \lambda^{i,A}$ satisfy the additive hazard model, i.e.\ that
        there are vector valued functions $h_t$ and $\tilde h_t$, and covariate processes $Z_t$ and $\tilde Z_t$ that are adapted to $\mathcal F_t^{i,\mathcal V_0 \cup \mathcal L} $  and $\mathcal F_t^{i,\mathcal V_0} $ respectively, and 
                \begin{equation}
                        \lambda_t^{i,A} = Y_t^{i,A} Z_t^{i \intercal} h_t  \text{ and } 
                        \tilde \lambda_t^{i,A} = Y_t^{i,A} \tilde Z_t^{i \intercal} \tilde
                        h_t. 
                \end{equation}
The previous equation translates into the following: 
\begin{align*}
	R_t^i &= R_0^i + \int_0^t R_{s-}^i dK_s^i \\
    K_t^i &= \int_0^t (\theta_s^i - 1) dN_s^{A,i} + \int_0^t Y_s^{i,A}Z_s^{i
            \intercal} dH_s - \int_0^t Y_s^{i,A}\tilde{Z}_s^{i \intercal}
    d\tilde{H}_s,
\end{align*}
where $H_t = \int_0^t h_s ds$ and  $\tilde H_t = \int_0^t \tilde h_s ds$. 
                
Our strategy is to replace $R_0^i$, $H$, $\tilde H$ and $\theta^i$ by estimators. This gives the following stochastic differential equation:   
\begin{align}
   {R}_t^{(i,n)} &= R_0^{(i,n)} + \int_0^t {R}_{s-}^{(i,n)}dK_s^{(i,n)} 
    \label{eq:Restimator} \\
    K_t^{(i,n)} &= \int_0^t ( \theta_{s-}^{(i,n)} - 1 )dN_s^{i,A} + 
    \int_0^t Y^{i,A}_s Z_{s-}^{i\intercal} d H^{(n)}_s - 
     \int_0^t Y^{i,A}_s \tilde Z_{s-}^{i \intercal} d \tilde H^{(n)}_s,
    \notag
\end{align}
where the quantity $R_0^{(i,n)}$ is assumed to be a consistent estimator of $R_0^i$. We will use the additive hazard regression estimators $H^{(n)}$ and $\tilde H^{(n)}$ for estimation $H$ and $\tilde H$, \cite{Andersen}. Moreover, suppose that  $\theta^{(i,n)}_0$ is a consistent estimator for $\theta_0^i$, the intensity ratio evaluated at zero. Our candidate for $\theta^{(i,n)}_t$, when $t > 0$, depends on the choice of an increasing sequence $\{\kappa_n\}_n$ with $\lim\limits_{n\longrightarrow \infty} \kappa_n = \infty$ such that $\sup_{n} \frac{\kappa_n }{\sqrt n } < \infty$. This estimator takes the form 
\begin{align}
    \theta^{(i,n)}_t &= \begin{cases}
            \theta_0^{(i,n)},  0 \leq  t < 1/ \kappa_n \\
            \frac{  \int_{t - 1/\kappa_n  }^t  Y_s^i \tilde Z_{s-}^{i\intercal } 
                    d\tilde H^{(n)}_s} {  \int_{t - 1/\kappa_n  }^t  Y_s^i  Z_{s-}^{i\intercal } 
                    d H^{(n)}_s}, 1/ \kappa_n    \leq t \leq T. 
    \end{cases}
    \label{eq:thetaTruncated}
\end{align}

$\kappa_n$ can thus be interpreted as a smoothing parameter. We let $Y^{(n),A}$ be the diagonal matrix where the $i$'th diagonal element is $Y^{i,A}$. The following Theorem says that the above strategy works out. 

\begin{theorem}\label{thm:ahwConsist}

        Suppose that 
        \begin{enumerate}[a.]

                \item \label{item:thetaConditions} Each $\theta^{i}$ is uniformly bounded, and right-continuous at $t = 0$. 
                \item For each $i$, \begin{equation} \label{eq:adhoc}
                                \lim_{ \delta \rightarrow 0} P \big( \inf\limits_{t \leq T} |
                        \tilde{Z}_t^{i\intercal} \tilde h_t | \leq \delta   \big)  =
                        0,
                \end{equation}
           \item \label{assumpt:supL2}  $E\big[\sup\limits_{s \leq T} |Z_{s}^i|^3_3 \big] < \infty $ and 
                 $E\big[\sup\limits_{s \leq T} |\tilde Z_{s}^i|^3_3 \big] < \infty $ for every $i$

         \item \label{assumpt:noColinearity} $$\lim_{a \rightarrow \infty} \sup_n 
         P\bigg( \sup_{s \leq T}  \Tr\Big( \big( \frac 1 n Z_{s}^{(n)\intercal } 
                 Y_{{s}}^{(n),A}Z_{s}^{(n)} \big)^{-1}\Big) \geq a   \bigg) = 0
                 $$ and 
                  $$\lim_{a \rightarrow \infty} \sup_n 
         P\bigg( \sup_{s \leq T}  \Tr \Big( \big( \frac 1 n  \tilde Z_{s}^{(n)\intercal } 
        Y_{{s}}^{(n),A} \tilde Z_{s}^{(n)} \big)^{-1}\Big) \geq a   \bigg) = 0
                 $$ 
 \end{enumerate}
Then we have that
\begin{equation}
        \lim_{n \rightarrow \infty} P \bigg( \sup_{t \leq T} |{R}_t^{(i,n)} - R^i_t
        | > \delta \bigg) = 0
\end{equation}
for every $\delta > 0$ and $i$.
\end{theorem}

For Theorem \ref{thm:AalenConsistency} to apply we need that our additive hazard weight estimator and the likelihood-ratio are uniformly bounded. The latter will for instance be the case if both $\lambda^{i,A}-\tilde{\lambda}^{i,A}$ and $\tilde{\lambda}^{i,A}/\lambda^{i,A}$ are uniformly bounded. We will, however, only assume that the theoretical weights $R^i$ are uniformly bounded. In that case we can also make our weight estimator $R^{(i,n)}$ uniformly bounded, by merely truncating trajectories that are too large.

      

\section{{Example}}
\label{section: Example}

\subsection{Software}
We have developed \texttt{R} software for estimation of continuous-time MSMs that solve ordinary differential equations, in which additive hazard models are used to model both the time to treatment and the time to the outcome of interest. Our procedure involves two steps: first, we estimate continuous-time weights using fitted values of the treatment model. These weights can be used to re-weight the sample for estimating the outcome model. Second, we take the cumulative hazard coefficients of the weighted (or causal) outcome model and transform them to estimate ODE parameters that have a more appealing interpretation than cumulative hazards. The two steps can be performed using the \texttt{R} packages \texttt{ahw} and \texttt{transform.hazards}, both of which are available in the repository \texttt{github.com/palryalen}. Below, we show an example on how to use the packages on simulated data.

\subsection{A simulation study}
We simulate an observational study where individuals may experience a terminating event $D$, so that the hazard for $D$ depends additively on the treatment $A$ and a covariate process $L$. $A$ and $L$ are counting processes that jump from 0 to 1 for an individual at the instant treatment is initiated or the covariate changes, respectively. The subjects receive treatment depending on $L$, such that $L$ is a time-dependent confounder. The subjects in the $L=1$ group can move into treatment, while the subjects in the $L=0$ group may receive treatment or move to the $L=1$ group in any order. All subjects are at risk of experiencing the terminating event. The following data generating hazards for $D, A,$ and $L$ are utilised:
\begin{align}
	\alpha_t^D &= \alpha_t^{D|0} + \alpha_t^{D|A}A_{t-} + \alpha_t^{D|L}L_{t-} + \alpha_t^{D|A,L} A_{t-} L_{t-} \label{eq:example outcome hazard} \\
    \alpha_t^A &= \alpha_t^{A|0} + \alpha_t^{A|L}L_{t-}  \label{eq:example treatment hazard}\\
    \alpha_t^L &= \alpha_t^{L|0} + \alpha_t^{L|A}A_{t-}. \notag 
\end{align}

We want to assess the effect of $A$ on $D$ we would see if $A$ were randomised, i.e.\ if treatment initiation did not depend on $L$. To find the effect $A$ has on $D$ we perform a weighted analysis.

We remark that this scenario could be made more complicated by e.g.\ allowing the subjects to move in and out of treatment, or have recurrent treatments. We could also have included a dependent censoring process, and re-weighted to a hypothetical scenario in which censoring were randomised (see Section \ref{sec:censoring weights}).

\subsection{Weight calculation using additive hazard models}
\label{section: ahw description}
We assume that the longitudinal data is organised such that each individual has multiple time-ordered rows; one row for each time either $A$, $L$ or $D$ changes.

Our goal is to convert the data to a format suitable for weighted additive hazard regression. Heuristically, the additive hazard estimates are cumulative sums of least square estimations evaluated at the event times in the sample. The main function will therefore need to do two jobs; a) the data must be expanded such that every individual, as long as he is still at risk of $D$, has a row for each time $D$ occurs in the population, and b) each of those rows must have an estimate of his weight process evaluated just before that event time.

Our software relies on the \texttt{aalen} function from the \texttt{timereg} package. We fit two additive hazard models for the transition from untreated to treated. The first model assesses the transitions that we observe, i.e.\ where treatment is influenced by a subjects realisation of $L$. Here, we use \eqref{eq:example treatment hazard}, i.e.\ the true data generating hazard model for treatment initiation; an additive hazard model with intercept and $L$ as a covariate. The second model describes the transitions under the hypothetical randomised trial in which each individual's treatment initiation time is a random draw of the treatment initiation times in the population as a whole. The treatment regime in our hypothetical trial is given by the marginal treatment initiation hazard of the study population, which is the hazard obtained by integrating out $L$ from \eqref{eq:example treatment hazard}. We estimate the cumulative hazard using the Nelson-Aalen estimator for the time to treatment initiation, by calling a marginal \texttt{aalen} regression.

In this way we obtain a factual and a hypothetical \texttt{aalen} object that are used as inputs in our \texttt{makeContWeights} function. Other input variables include the bandwidth parameter used in \eqref{eq:thetaTruncated}, weight truncation options, and an option to plot the weight trajectories.

The output of the \texttt{makeContWeights} function is an expanded data frame where each individual has a row for every event time in the population, with an additional \texttt{weight} column containing time-updated weight estimates. To do a weighted additive hazard regression for the outcome, we will use the \texttt{aalen} function once again. Weighted regression is performed on the expanded data frame by setting the weights argument equal to the weight column.

When the weighted cumulative hazard estimates are at hand, we can transform our cumulative hazard estimates as suggested in Section \ref{section:transforming hazards}, to obtain effect measures that are easier to interpret. This step can be performed using the \texttt{transform.hazards} package; see the GitHub vignette for several worked examples.

\subsection{A marginal structural model}
\label{section:a marginal structural model}
We now suppose the intervention that imposes a marginal treatment initiation rate is causally valid. This implies that the intensity for the event $D$ has the same form under the randomised scenario $\tilde P$, i.e.\ that the hazard for $D$ under $\tilde P$ for the filtration $\mathcal{F}_t^{A \cup D \cup L}$, generated by $A,D,$ and $L$, takes the same functional form as \eqref{eq:example outcome hazard}. We are, however, interested in the hazard with respect to $\tilde P$ and the subfiltration $\mathcal{F}_t^{A \cup D}$, the filtration generated by $A$ and $D$ (note that $\mathcal{F}_t^{A \cup D \cup L}$ and $\mathcal{F}_t^{A \cup D}$ respectively correspond to $\mathcal{F}_t^{\mathcal{V}_0 \cup \mathcal{L}}$ and $\mathcal{F}_t^{\mathcal{V}_0}$ from Section \ref{section:Hypothetical scenarios and likelihood ratios}). By the innovation theorem the hazard with respect to $\tilde P$ and $\mathcal{F}_t^{A \cup D}$ takes the form
\begin{align*}
    \beta(t | A) = \beta_t^0 + \beta_t^A A_{t-}.
\end{align*}
A straightforward regression analysis of the observational data cannot yield causal estimates. Using the ideas from Section \ref{section: Derivations of additive hazards}, we can estimate the cumulative coefficients $B_t^{A=0} = \int_0^t \beta_s^0 ds$ and $B_t^{A=1} - B_t^{A=0} = \int_0^t \beta_s^A ds$ consistently by performing a weighted additive hazard regression.

Cumulative hazards, however, are not easy to interpret. We therefore assess effects on the survival scale, using a marginal structural relative survival model. In this example, our marginal structural relative survival $RS^A$ solves 
\begin{align}
	RS_t^{A=a} = 1 + \int_0^t \begin{pmatrix}
		-RS_{s}^{A=a} & RS_{s}^{A=a}
	\end{pmatrix}
    d\begin{pmatrix}
    	B_s^{A=a} \\ 
        B_s^{A=0}
    \end{pmatrix}. \label{eq: MSM relative survival}
\end{align}
The quantity $RS^{A=1}$ can be understood as the survival probability a subject would have if he were exposed at time 0, relative to the survival probability he would have if he were never exposed. Our suggested plugin-estimator is obtained by inserting the estimated causal cumulative coefficients, i.e.\ the weighted estimates $\hat B^{A=a}$ and $\hat B^{A=0}$:
\begin{align*}
	\hat{RS}_t^{A=a} &=  1 + \int_0^t \begin{pmatrix}
		-\hat{RS}_{s-}^{A=a} & \hat{RS}_{s-}^{A=a}
	\end{pmatrix}
    d\begin{pmatrix}
    	\hat B_s^{A=a} \\ 
        \hat B_s^{A=0}
    \end{pmatrix}.
\end{align*}

\subsection{Simulation details and results}

We simulate subjects, none of which are treated at baseline. Initially, all the patients start with $L = 0$, and the hazards for transitioning from one state to another is constant. As described in Section \ref{section: ahw description}, we fit additive hazard models for the time to treatment initiation, one for the observed treatment scenario, i.e.\ \eqref{eq:example treatment hazard}, and one for the hypothetical randomised scenario. These models are inserted into \texttt{makeContWeights} to obtain weight estimates. Finally, we estimate the additive hazard model by calling the \texttt{aalen} function where the \texttt{weights} option is set equal to the weight column in the expanded data set.

We make comparisons to the discrete-time, stabilised IPTWs, calculated using pooled logistic regressions. To do this, we discretise the study period $[0,10]$ into $K$ equidistant subintervals, and include the time intervals as categorical variables in the regressions. We fit two logistic regressions; one for the weight numerator, regressing only on the intercept and the categorical time variables, and a covariate-dependent model for the weight denominator, regressing on the intercept, the categorical time variables, and the time-updated covariate process. We then calculate IPTWs by extracting the predicted probabilities of the two logistic regression model fits, and inserting them into the cumulative product formula \cite[eq. (17)]{robins2000marginal}. 

In the upper three rows of Figure \ref{fig:treatmentEffect} we display estimates of the causal cumulative hazard coefficient, i.e.\ estimates of $B^{A=1} - B^{A=0}$, for a range of sample sises. We include estimates weighted according to our estimator \eqref{eq:Restimator}, the IPTW estimator, and the theoretical weights, i.e.\ the true likelihood ratios $\{ R^i\}_i$. Compared to the discrete weight estimators, our continuous-time weight estimator \eqref{eq:Restimator} gives better approximations to the curves that are estimated with the theoretical weights. In the lowest row of Figure \ref{fig:treatmentEffect} we plot $\hat {RS}^{A=1}$, i.e.\ transformed estimates of the cumulative hazard coefficients re-weighted according to the different weight estimators. We used the \texttt{transform.hazards} package to perform the plugin-estimation.




      \begin{figure}
\setlength{\lineskip}{1ex}
\subfloat{\includegraphics[width=.3\textwidth]{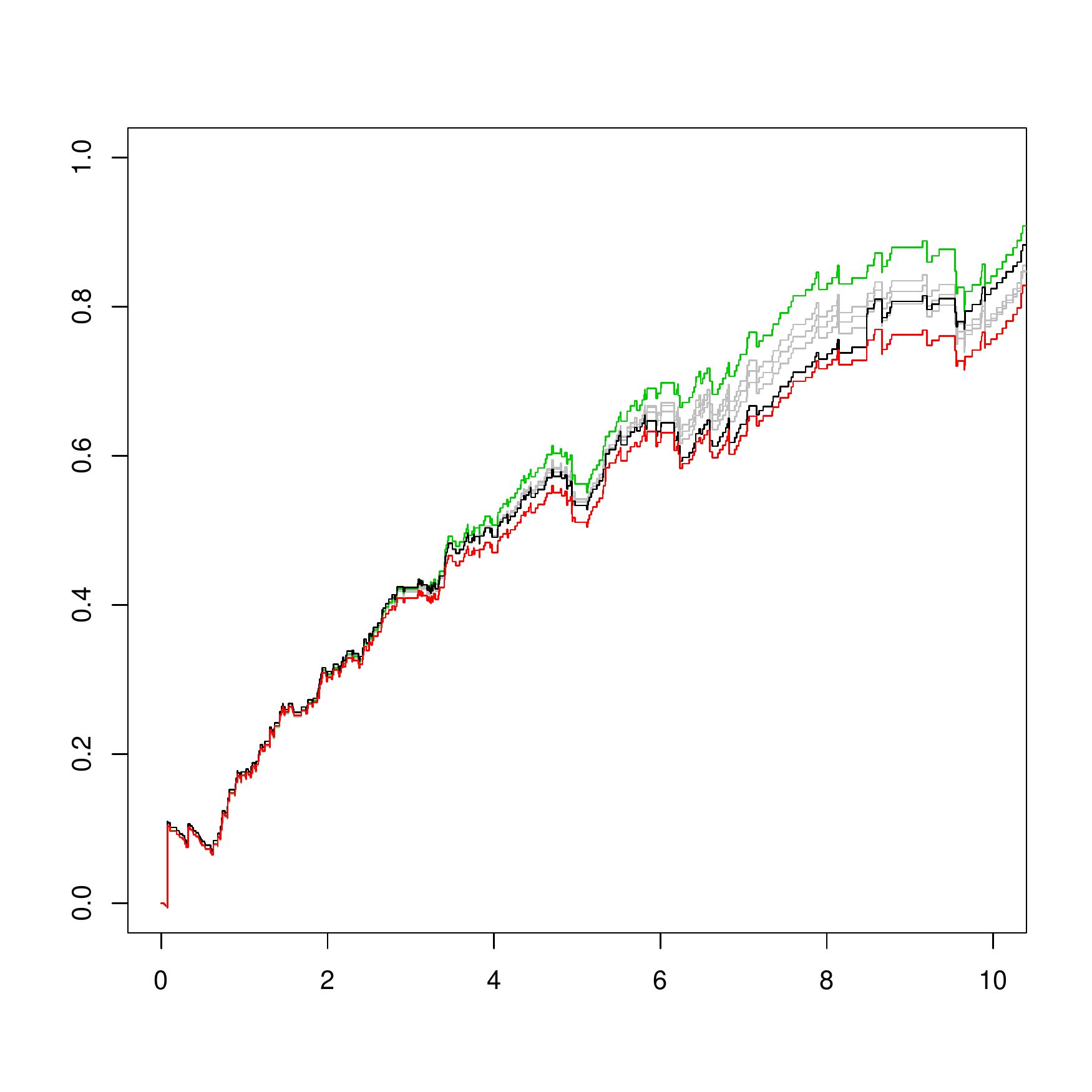}}%
\hspace{.03\textwidth}%
\subfloat{\includegraphics[width=.3\textwidth]{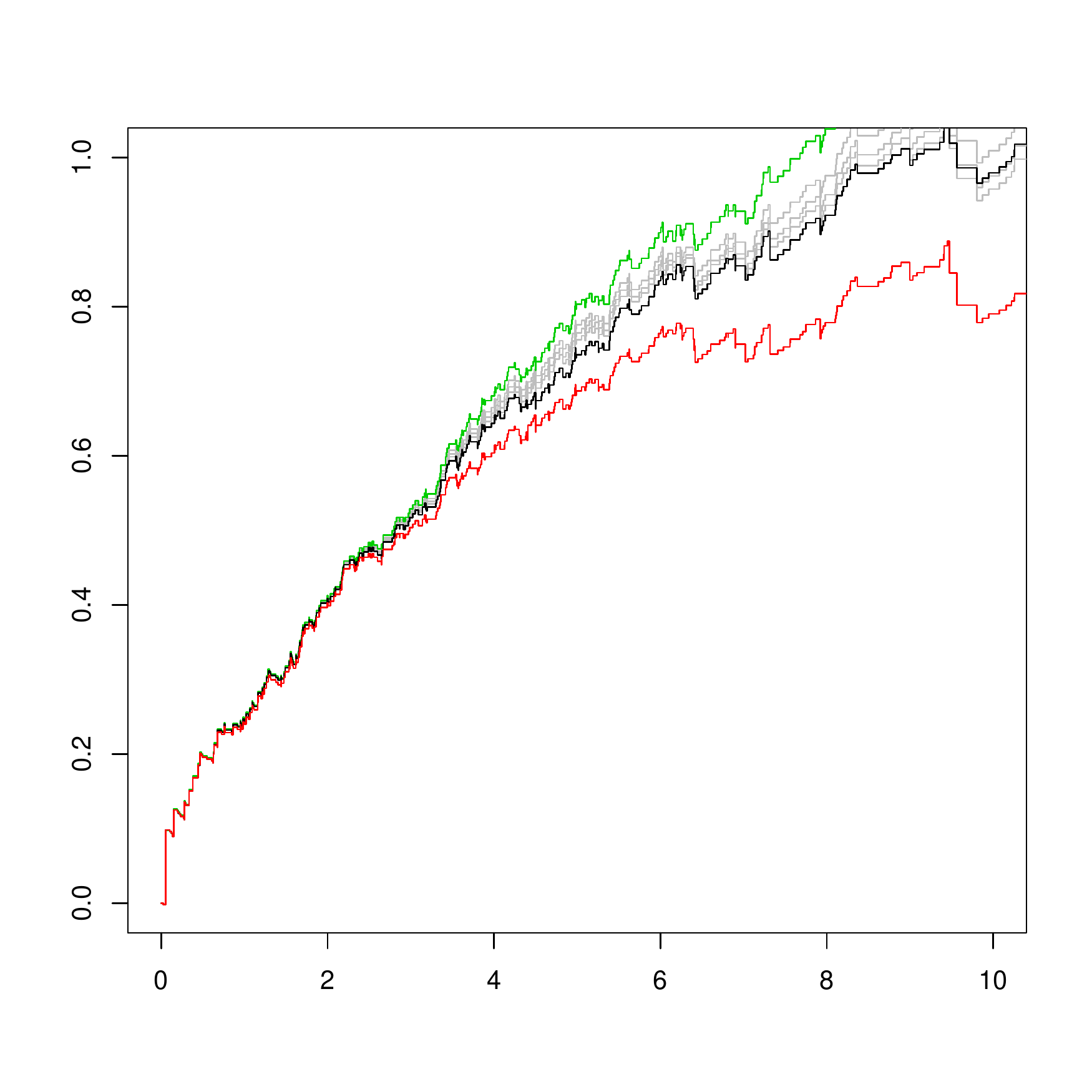}}%
\hspace{.03\textwidth}%
\subfloat{\includegraphics[width=.3\textwidth]{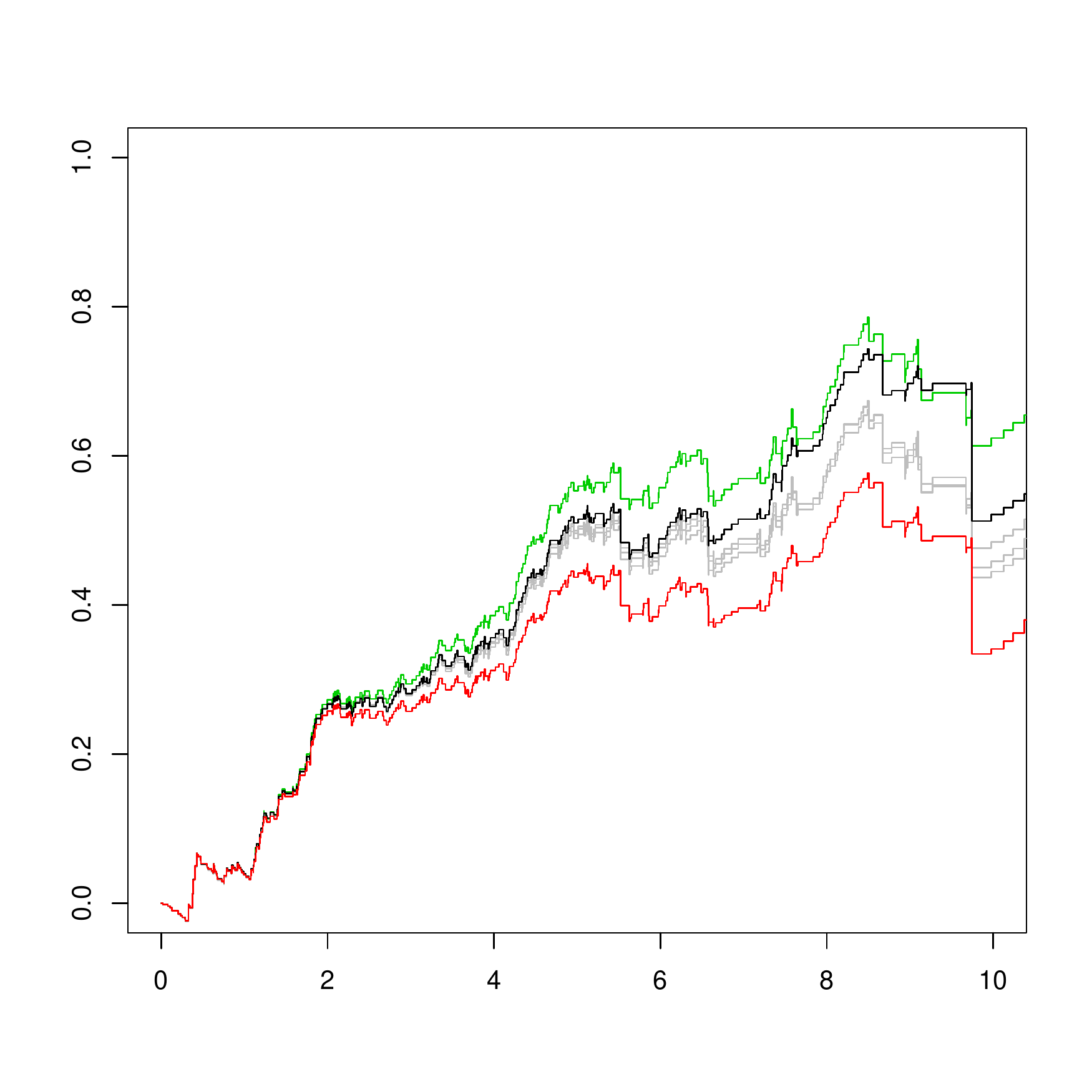}}%
\hspace{.03\textwidth}%
\subfloat{\includegraphics[width=.3\textwidth]{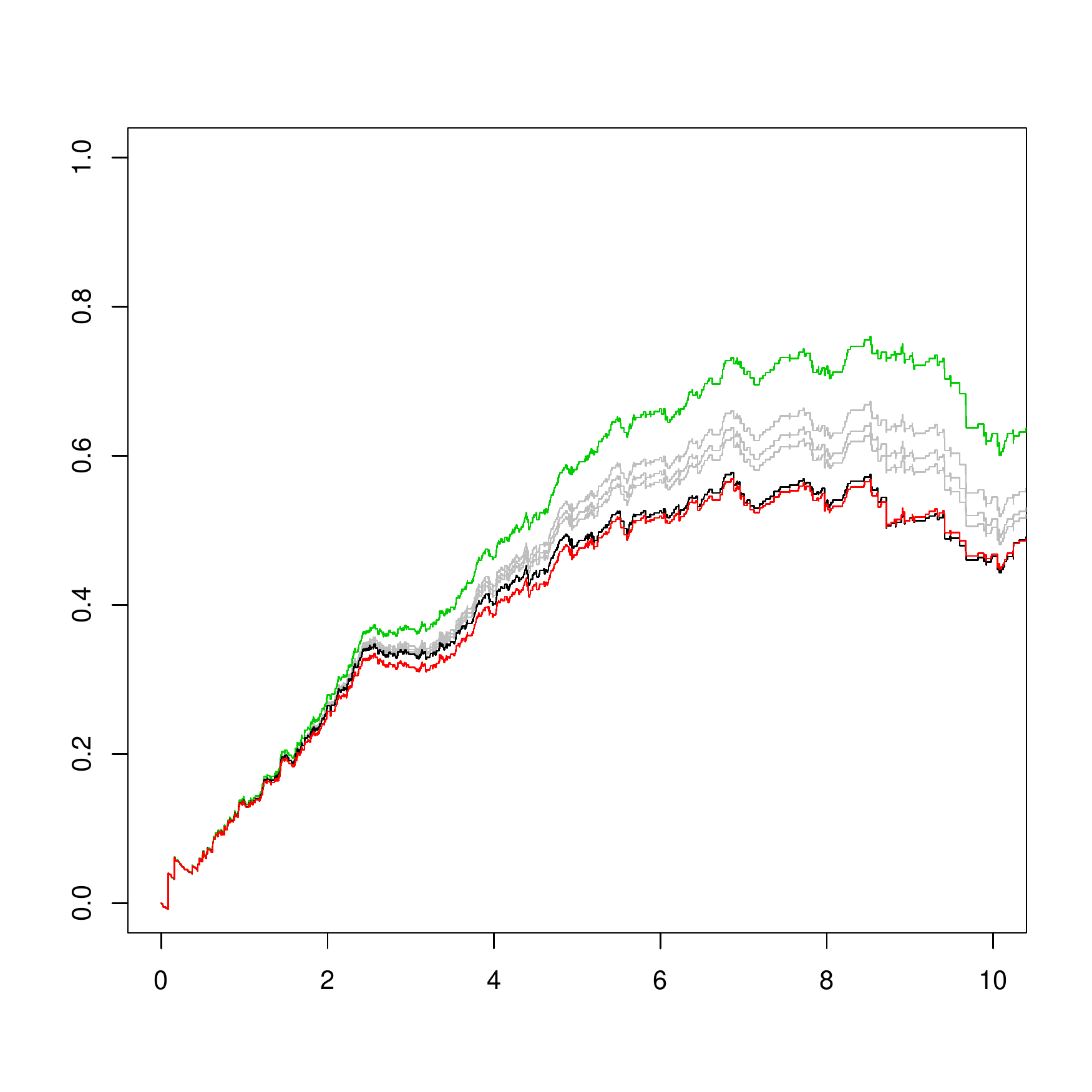}}%
\hspace{.03\textwidth}%
\subfloat{\includegraphics[width=.3\textwidth]{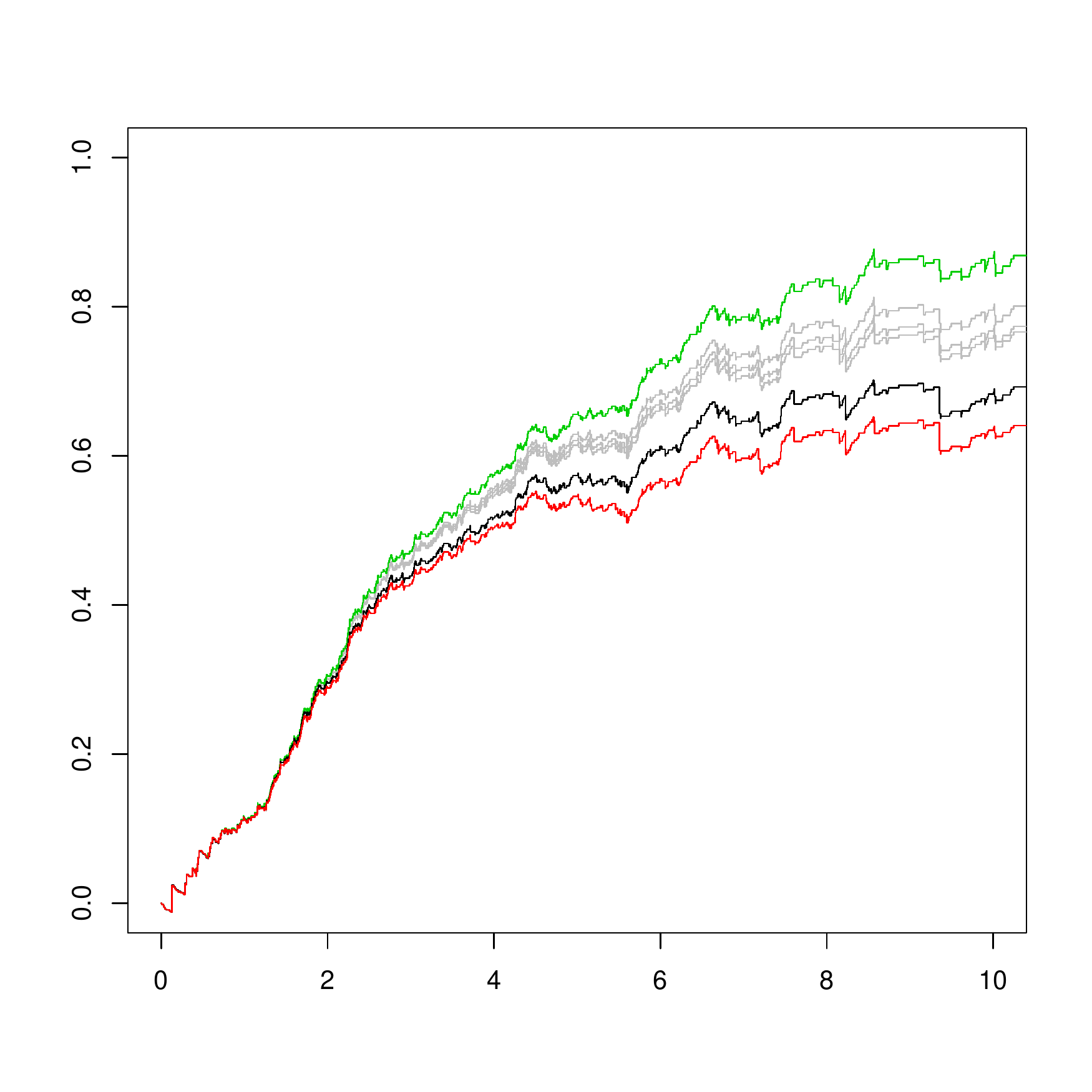}}%
\hspace{.03\textwidth}%
\subfloat{\includegraphics[width=.3\textwidth]{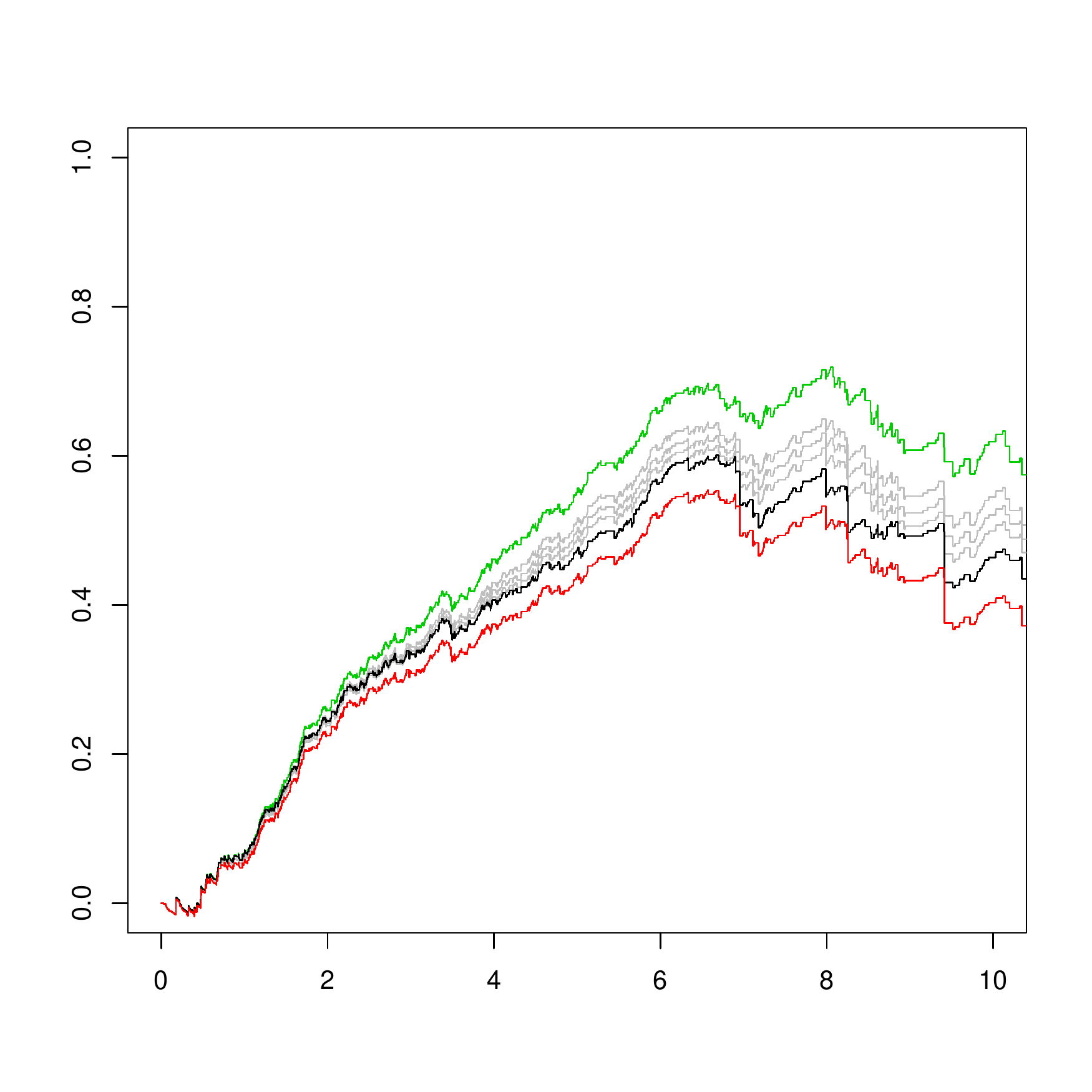}}%
\hspace{.03\textwidth}%
\subfloat{\includegraphics[width=.3\textwidth]{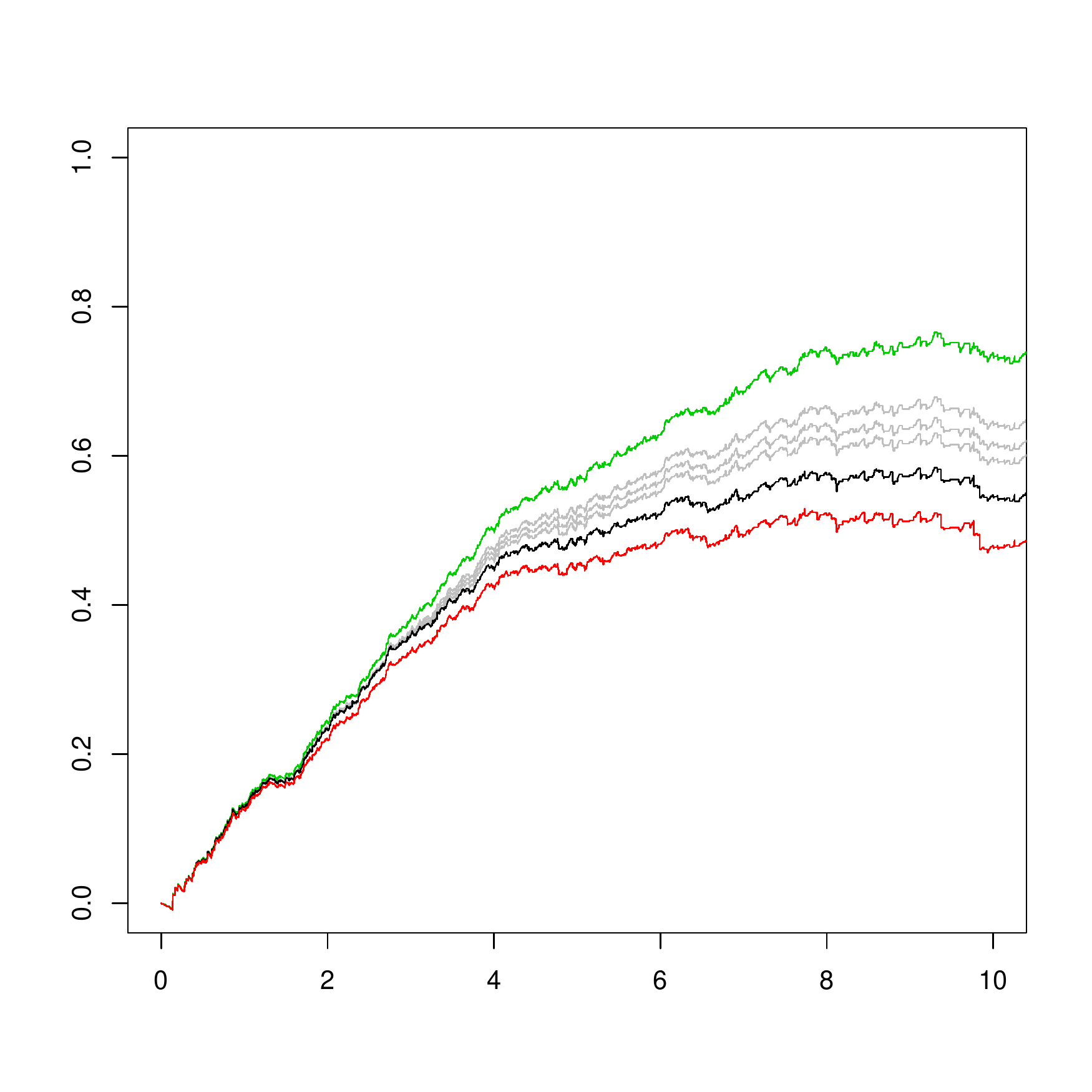}}%
\hspace{.03\textwidth}%
\subfloat{\includegraphics[width=.3\textwidth]{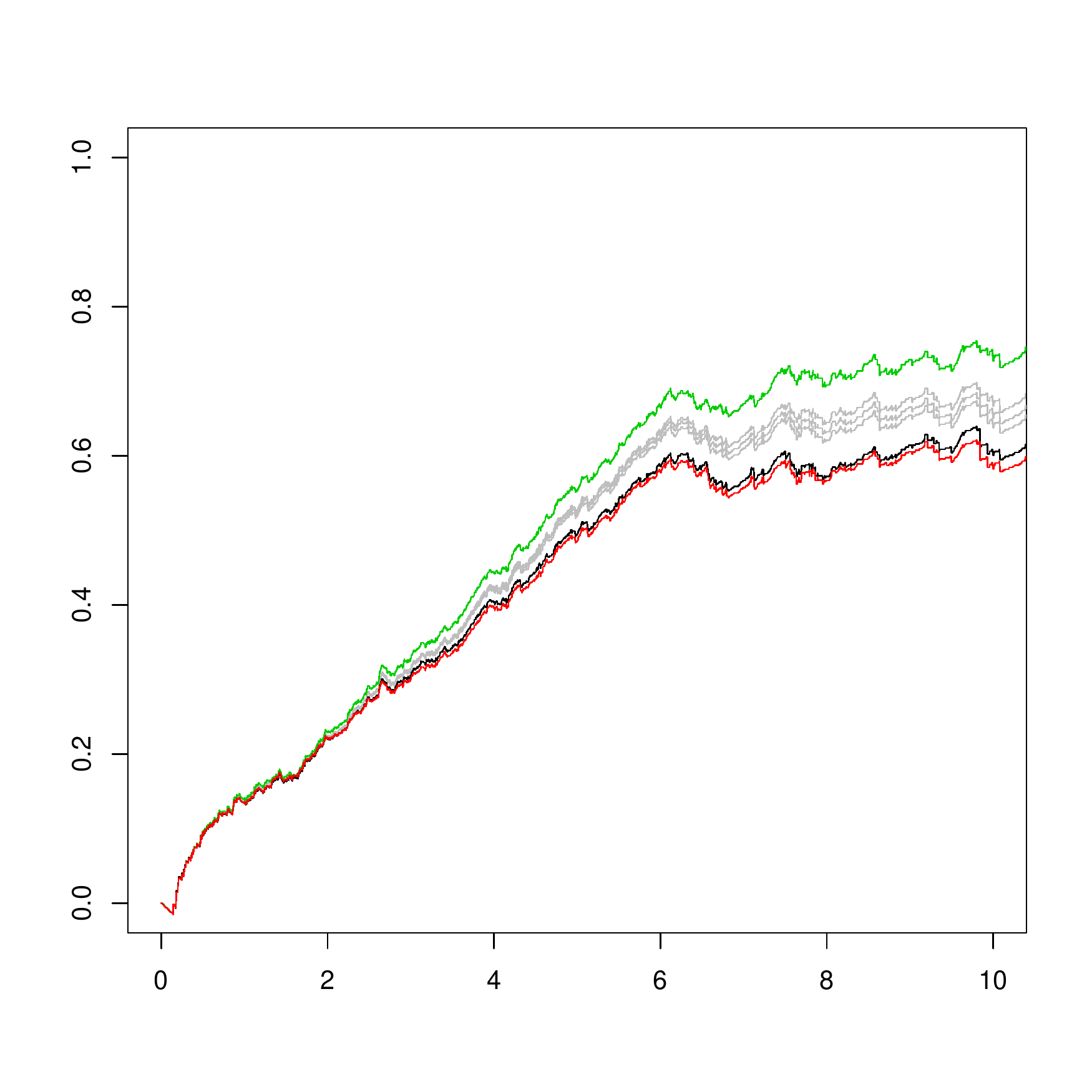}}%
\hspace{.03\textwidth}%
\subfloat{\includegraphics[width=.3\textwidth]{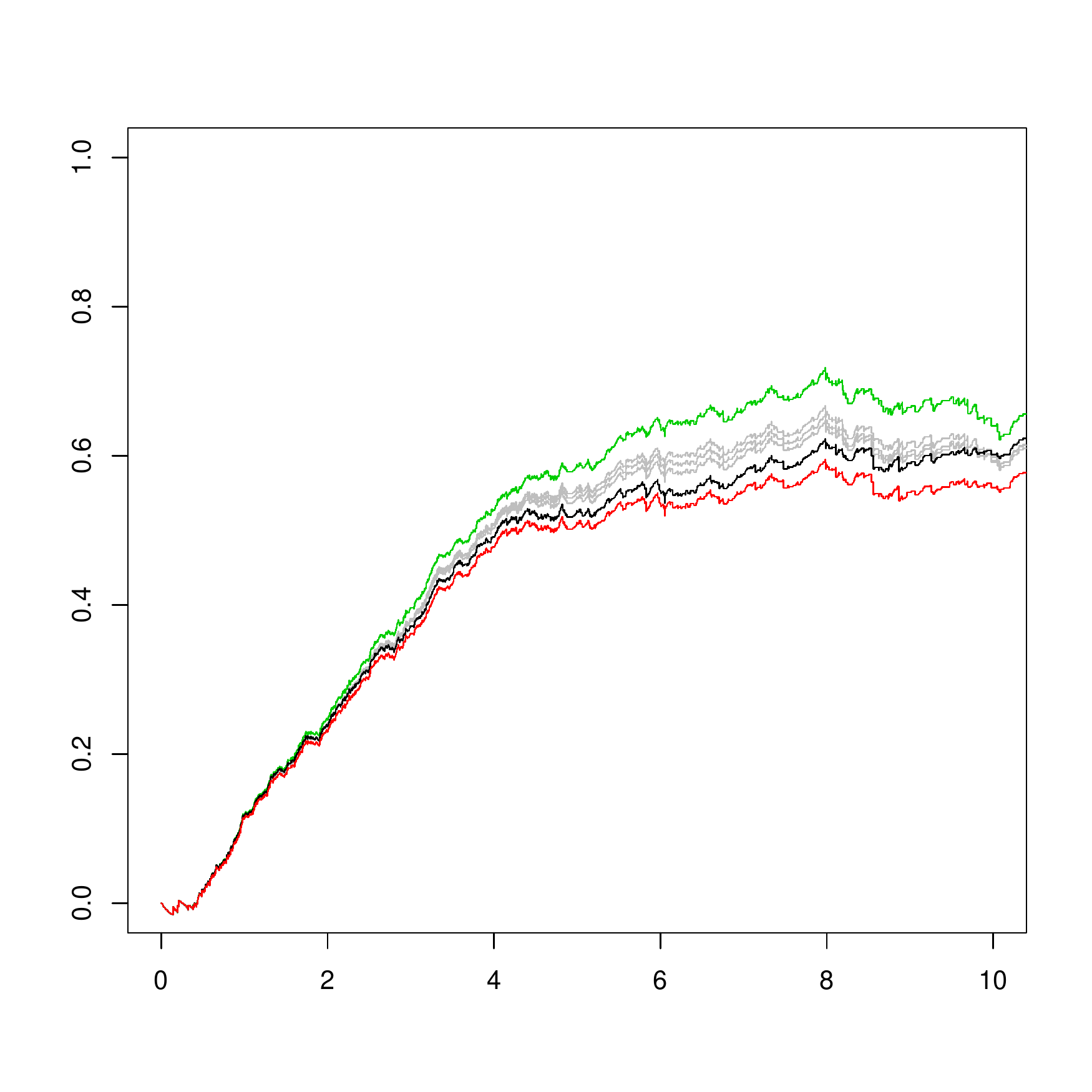}}%
\hspace{.03\textwidth}%
\subfloat{\includegraphics[width=.3\textwidth]{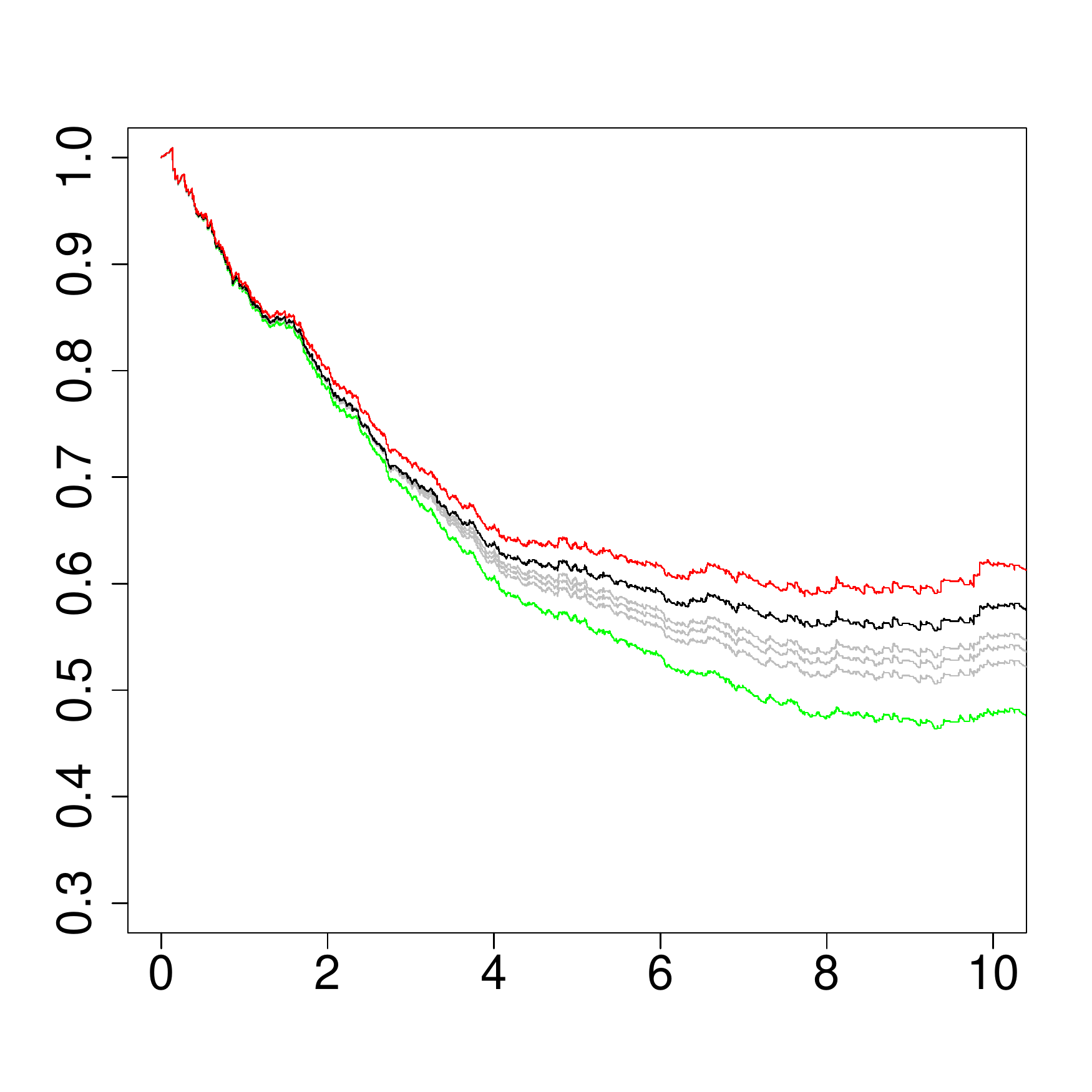}}%
\hspace{.03\textwidth}%
\subfloat{\includegraphics[width=.3\textwidth]{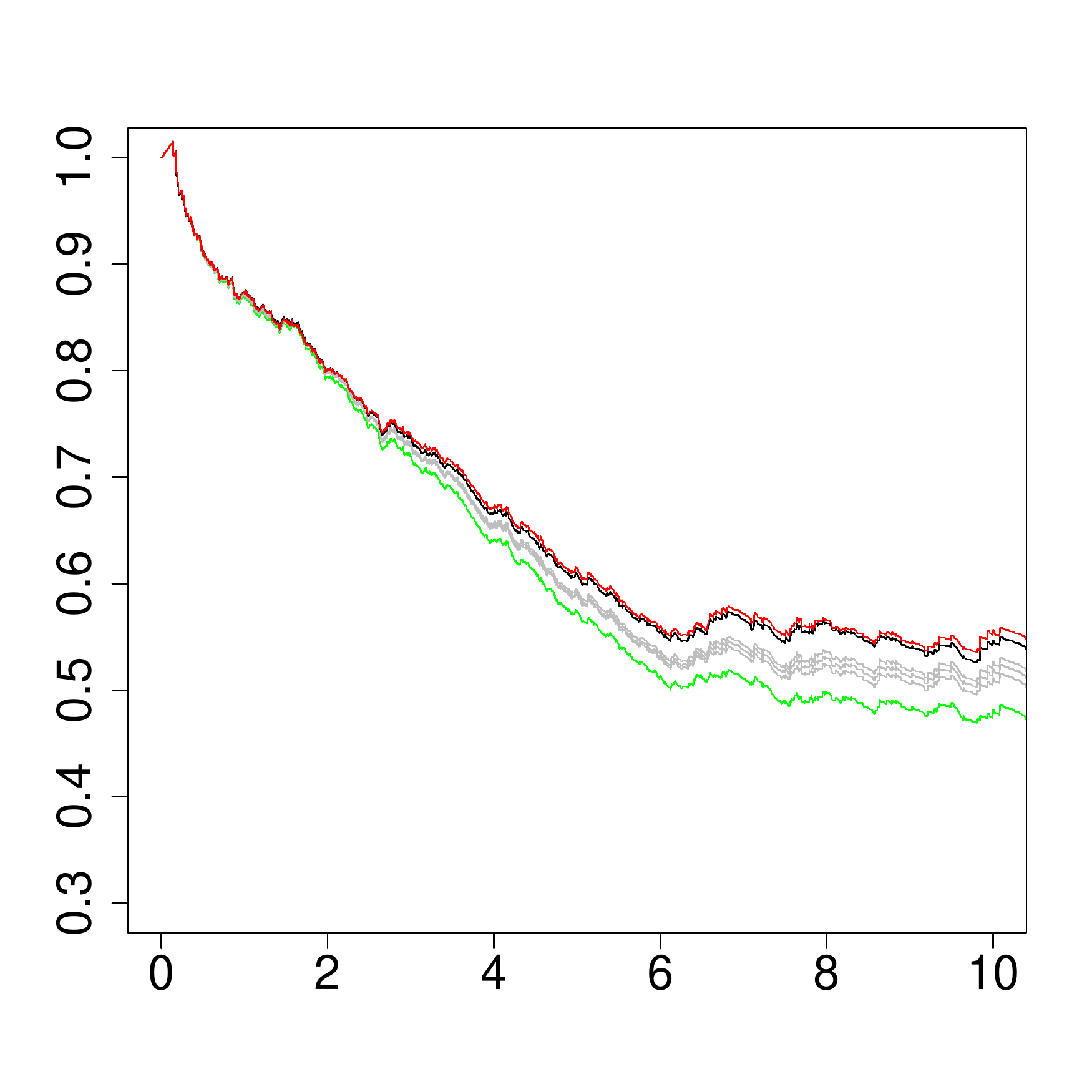}}%
\hspace{.03\textwidth}%
\subfloat{\includegraphics[width=.3\textwidth]{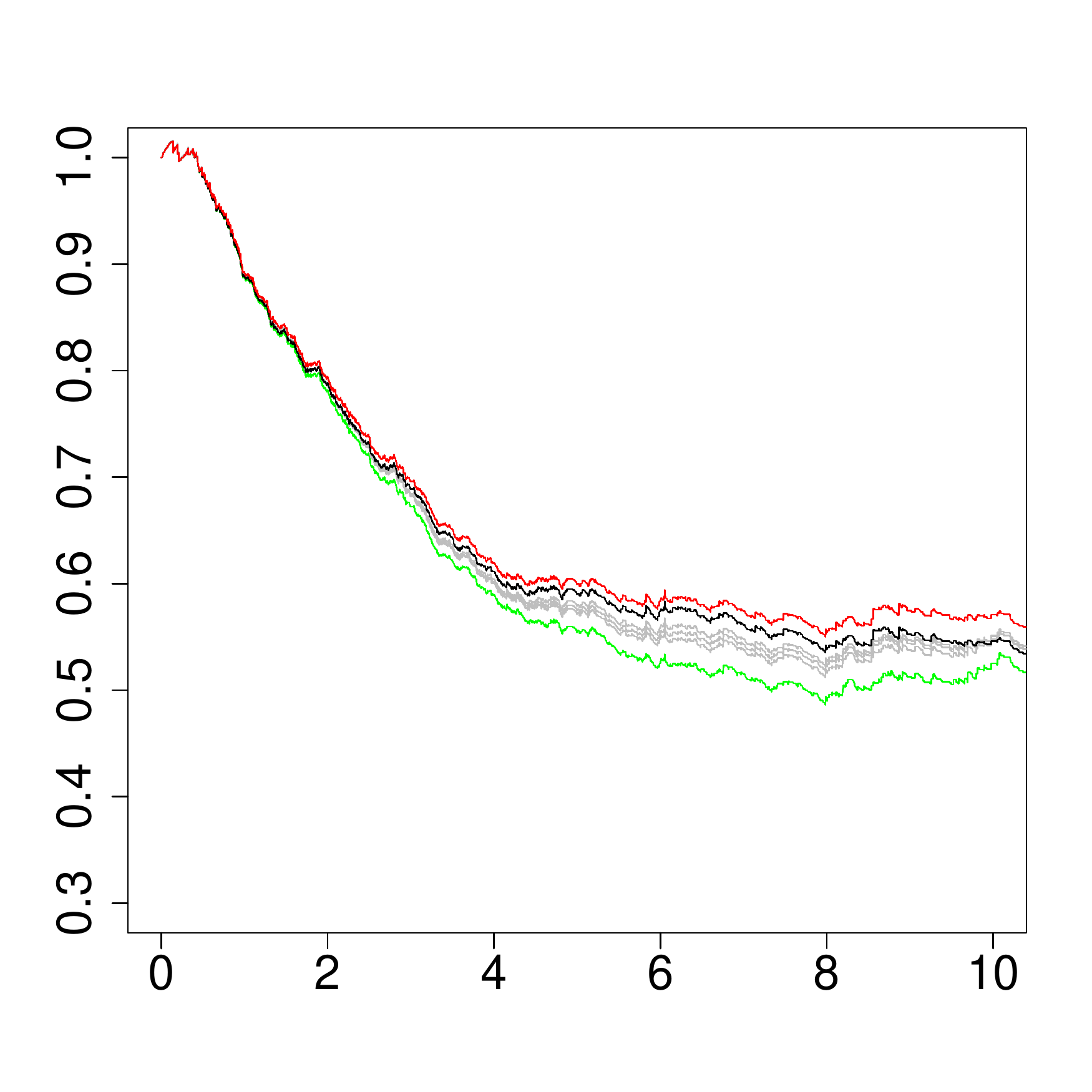}}%
\hspace{.03\textwidth}%
 \caption{The upper three rows: Three realisations of the cumulative treatment effect estimates for the same scenario, with $n=$ 500, 1000, and 2000 from top to bottom. A red line based on estimates re-weighted with the true $R^i$'s is included for reference. The green line shows the unweighted estimates, the gray lines are obtained using the IPTW estimates, while the black line is obtained using our additive hazard weight estimates. The discrete weights were estimated using pooled logistic regressions based on $K = $ 4, 8, and 16 time intervals. Increasing the number of intervals moved the curves closer to the red curve. The lowermost row: Estimated causal effect of being treated at $t=0$ versus never being treated according to the relative survival MSM, based on the $n=2000$ sample.}
 \label{fig:treatmentEffect}
\end{figure}

 \begin{figure}
\includegraphics[width=\textwidth]{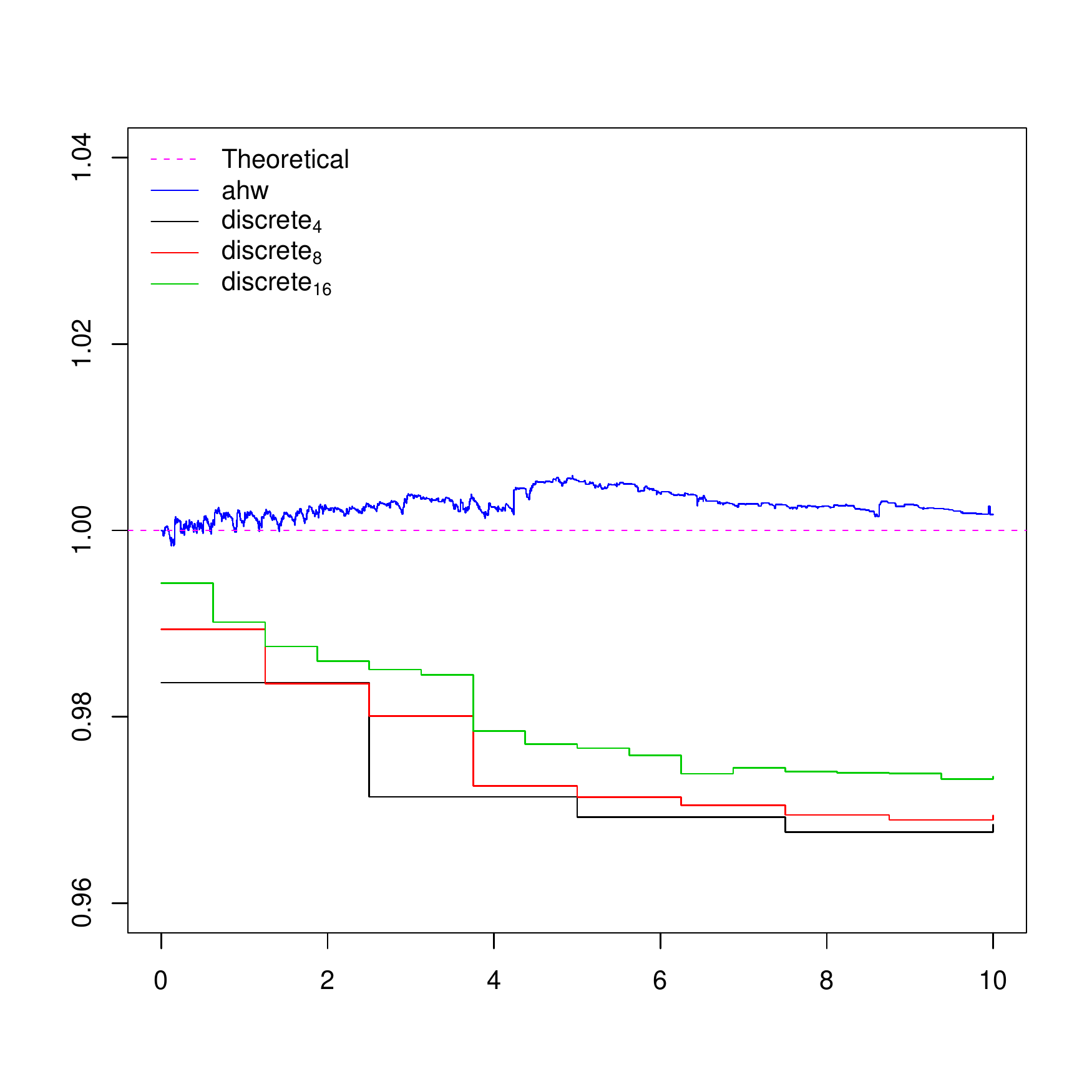}
 \caption{Average weights based on a sample size of 3000. The theoretical weights have expected value 1. Included are our additive hazard weights, as well as IPTW with $K = $ 4, 8, and 16 time intervals. We see that the discrete weights are biased approximations of the theoretical likelihood ratio, while our additive hazard weight estimator appears to be less biased.}
 \label{fig:discreteWeightBias}
\end{figure}

\section{Performance}

In Figure \ref{fig:discreteWeightBias} we plot mean weight estimates based on aggregated simulations of the set-up in Section \ref{section: Example}. The plot suggests that the discrete weights gradually approximate the continuous likelihood ratio as the time discretisation is refined. However, the continuous-time weights \eqref{eq:Restimator} are closer to the expected value of 1 at all times $t$, indicating less bias.

Choosing the bandwidth parameter will influence the weight estimator and weighted additive hazard estimator in a bias-variance tradeoff; a small $\kappa_n$ will yield estimates with large bias and small variance, while a large $\kappa_n$ will give rise to small bias but large variance. It is difficult to provide an exact recipe for choosing the bandwidth parameter, since a good choice depends on several factors, such as the sample size, the distribution of the treatment times, as well as the form and complexity of the true treatment model: If the true treatment hazard is constant, a small $\kappa_n$ is often appropriate. If the treatment hazard is highly time-varying, $\kappa_n$ should be chosen to be large, depending on the sample size. Heuristically, several treatment times in the interval $[t-1/\kappa_n, t]$ for each $t$ would be desirable, but this is not possible in every situation, e.g.\ when the treatment time distribution is skewed. Such distributions can lead to instable, and possibly large weights for some subjects, even if the chosen bandwidth parameter is a good choice for most other subjects. One option is to truncate weights that are larger than a specified threshold, at the cost of introducing bias. We can assess sensitivity concerning the choice of the bandwidth by performing an analysis for several bandwidth values, truncating weights if necessary, and comparing the resulting weighted estimators. This approach was taken in \cite[see e.g.\ Supplementary Figure 4]{ryalen2018pcancer}, where no noticeable difference was found for four values of $\kappa_n$.

We inspect the bias and variance of our weight estimator for sample sizes $n$ under four bandwidth choices $\kappa_n^z$, $z=1,2,3,4$ at a specified time $t_0$. By aggregating estimates of $k$ samples for each $n$ we get precise estimates of the bias and variance as a function of $n$ for each choice. The bandwidth functions are scaled such that they are identical at the smallest sample $n_0$, with $\kappa_{n_0}^1 = \kappa_{n_0}^2 = \kappa_{n_0}^3 = \kappa_{n_0}^4 = 1/t_0$. Otherwise they satisfy $\kappa_n^1 \propto n^{1/2}, \kappa_n^2 \propto n^{1/3},    \kappa_n^3 \propto n^{1/5},$ and $\kappa_n^4 \propto n^{1/10}.$

We simulate a simple scenario where time to treatment initiation depends on a binary baseline variable, such that $\lambda_t^{i,A} = Y_t^{i,A} ( \alpha_t^{0} + \alpha_t^A x^i)$ for individual $i$ with at-risk indicator $Y^{i,A}$ and binary variable $x^i$. We calculate weights that re-weight to a scenario where the baseline variable has been marginalised out, i.e.\ where the treatment initiation intensity is marginal. Utilising the fact that the true likelihood ratio $R^i$ has a constant mean equal to 1, we can find precise estimates of the bias and variance of the additive hazard weight estimator \eqref{eq:Restimator} at time $t_0$.

We plot the bias and variance of the weight estimator as a function of $n$ under the strategies $\kappa_n^1, \kappa_n^2, \kappa_n^3$ and $\kappa_n^4$ in Figure \ref{fig:biasvariance}. We see that the convergence strategy $\kappa^1_n$ yields a faster relative decline in bias, but a higher variance as the sample size increases. Meanwhile, the strategy $\kappa^4_n$ has a slower decline in bias, but a smaller variance than the other strategies. Finally, the strategies $\kappa_n^2$ and $\kappa_n^3$ lie mostly between $\kappa_n^1$ and $\kappa_n^4$ both concerning bias and variance, as a function of the sample size. We also see empirical justification for the requirement $\sup_n \kappa_n/n^{1/2} < \infty$, as the variance under the strategy $\kappa_n^1$ declines very slowly as $n$ is increased.

\begin{figure}
\centering
\setlength{\lineskip}{1ex}
\subfloat{\includegraphics[width=.6\textwidth]{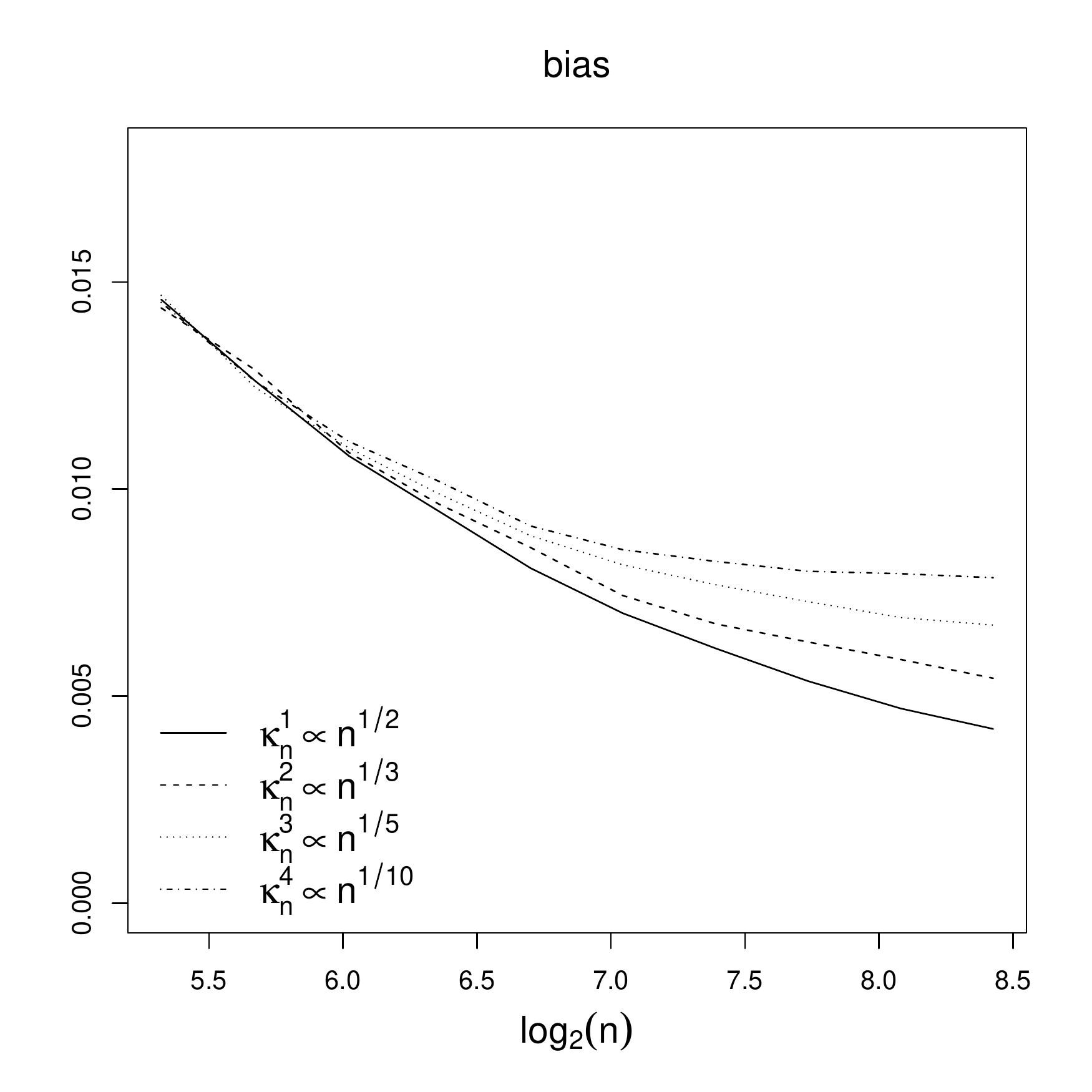}}%
\subfloat{\includegraphics[width=.6\textwidth]{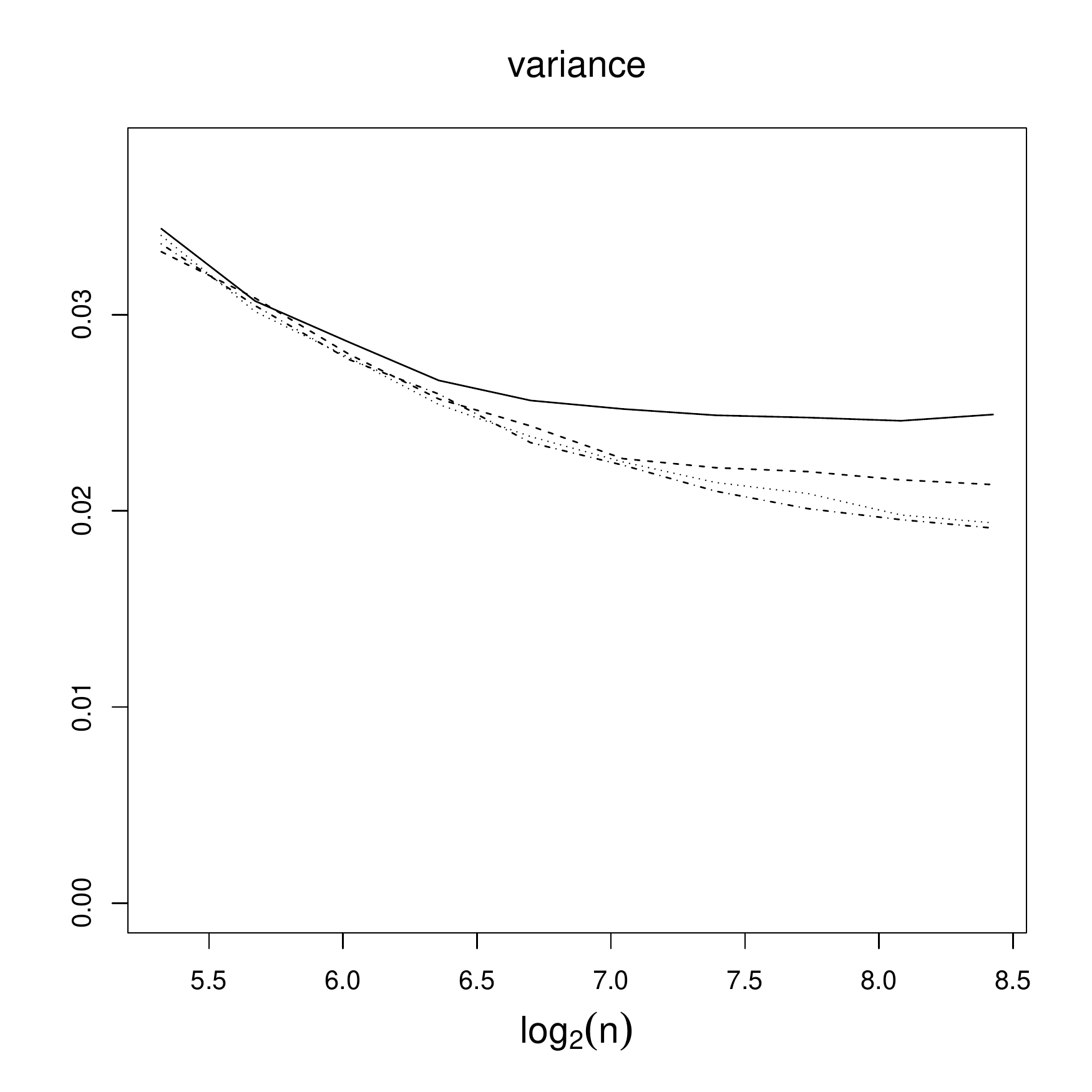}}%
 \caption{Bias and variance as a function of $n$, for four bandwidth refinement strategies.}
\label{fig:biasvariance}
\end{figure}

\section{Censoring weights}\label{sec:censoring weights}

                            Most standard martingale-based estimators in survival
                      analysis are consistent when we have independent censoring, see \cite[III.2.1]{Andersen}. 
              We have assumed independent censoring when
                      conditioning on $\mathcal V_0$. 
              A likely situation where this is violated is when we have  independent censoring when conditioned on  $\mathcal L \cup \mathcal V_0$, but have
                      dependent censoring if we only condition on $\mathcal V_0$.  
              If the model is causal with respect to an intervention
                      that randomises censoring sufficiently, we can model the scenario where this intervention had been applied, and censoring is independent when conditioning on $\mathcal V_0$. 
                 This means that many estimators that are common in
                         survival analysis will be consistent. 
                 Suppose that $N^{i,c}$ is a counting process that jumps when
                         individual $i$ is censored. 
                 Moreover, let $\lambda_t^{i,c}$ denote the 
                         intensity of $N^{i,c}$ 
                         with respect to the filtration $\mathcal F_t
                         ^{i,\mathcal V_0 \cup \mathcal L}$, and let 
                          $\tilde \lambda_t^{i,c}$ denote its 
                         intensity of 
                         with respect to the filtration $\mathcal F_t
                         ^{i,\mathcal V_0}$.

                 Suppose that there is a meaningful intervention that 
                       would give a scenario with frequencies that are governed 
                       by $\tilde P$ and its intensity for censoring with
                       respect to  $\mathcal F_t
                         ^{i, \mathcal V_0 \cup \mathcal L}$, is replaced by 
                       $\tilde \lambda^{i,c}_t$.
                 If the model is causal with respect to this intervention, the corresponding 
                         likelihood ratio process is given by
                         \begin{equation}
                                 R_t^{i,c}  = \prod_{s  \leq t} 
                                \Big( \frac{\tilde \lambda_s^{i,c}}
                                {\lambda_s^{i,c} }\Big)^{\Delta N_s^{i,c} }
                                 \exp \big( -\int_0^t \tilde \lambda_s^{i,c} -
                                 \lambda_s^{i,c}  ds  \big). 
                         \end{equation}
                 However, as we only need to apply weights to
                         observations strictly before the time of censoring, we
                         only need to consider
                         \begin{equation}
                                 R^{i,c}_t = \exp \big( -\int_0^t \tilde
                                 \lambda_s^{i,c} -
                                 \lambda_s^{i,c}  ds  \big). 
                         \end{equation}
                 This process is a solution to the equation
                         \begin{equation}
                                 R^{i,c}_t = 1 + \int_0^t   R_{s}^{i,c}(
                                 \lambda_s^{i,c} -
                                 \tilde \lambda_s^{i,c})  ds.  
                         \end{equation}
                 Furthermore, we assume additive hazard models, i.e.\ that
                 \begin{equation}
                         \lambda_t^c = Y^{i,c}  U_{s-
                         }^{i\intercal} g_t 
                         \text{ and }  
                          \tilde \lambda_t^{i,c} = Y^{i,c}
                          \tilde U_{s-}^{i\intercal}
                           \tilde g_t,
                 \end{equation}
                 for an $\mathcal F_t^{i,\mathcal V_0 \cup \mathcal L}$-adapted covariate process 
                 $U_t^i $, and 
                 an $\mathcal F_t^{i,\mathcal V_0 }$-adapted covariate process 
                 $\tilde U_t^i $, and vector valued functions $g$ and $\tilde g$. 
         Following Theorem \ref{thm:ahwConsist}, we see that these
                 weights are consistently estimated by $R^{(i,n,c)}$ defined by
                 the equation: 
                 \begin{align*}
                      R^{(i,n,c)}_t & = 1 + \int_0^t    R^{(i,n,c)}_{s-} 
                      dK^{(i,n,c)}_s \\ 
                    K^{(i,n,c)}_t & = \int_0^t   Y^{i,c}_s U_{s-}^{i
                            \intercal} dG_s^{(n)} - 
                     \int_0^t   Y^{i,c}_s \tilde U_{s-}^{i
                            \intercal} d\tilde G_s^{(n)},
                 \end{align*}
                 where $G^{(n)}$ and $\tilde G^{(n)}$ are the usual additive
                 hazards estimates of $\int_0^\cdot g_s ds$ and $\int_0^\cdot \tilde{g}_s ds$.

\section{Discussion}\label{sec:discussion}
Marginal structural modeling is an appealing concept for causal survival analysis. Here we have developed theory for continuous-time MSMs that may motivate the approach for practical research. Indeed, we show that the continuous-time MSMs yield consistent effect estimates, even if the treatment weights are estimated from the data. Our continuous-time weights seem to perform better than the discrete time weights when we study processes that develop in continuous time. Furthermore, our weights can be estimated using additive hazard regressions, which are easy to fit in practice. Importantly, we also show that causal effect estimates on the hazard scale, e.g.\ weighted cumulative hazard estimates, can be transformed consistently to estimate other parameters that are easier to interpret causally. We thereby offer a broad strategy to obtain causal effect estimates for time-to-event outcomes. Previously, \cite{McKeague1} and \cite{McKeague2} derived results on weighted additive hazard regression, but they do not cover our needs, as our weights are estimates of likelihood ratios with respect to filtrations that are larger than the filtration for the additive hazard that we want to estimate. 


Estimators of IPTWs may be unstable and inefficient, e.g.\ when there are strong predictors of the treatment allocation. In practice, applied researchers will often face a bias-variance tradeoff when considering confounder control and efficient weight estimation. This bias-variance tradeoff has been discussed in the literature, and weight truncation has been suggested to reduce the variance, at the cost of introducing bias; see e.g.\ \cite{cole2008constructing}. Similar to IPTWs, and for the same reasons, our continuous-time weight estimator may be instable, and proper weight estimation requires a delicate balance between confounder control and precision in most practical situations.

We have considered the treatment process $A$ to be a time-to-event variable, but our strategy can be generalised to handle recurrent, or piecewise constant exposures. If $A$ is allowed to have multiple jumps, the estimation procedure becomes more complex, but the same estimators \eqref{eq:weighted cumulative hazard estimator} and \eqref{eq:Restimator} can be used with few modifications. We think, however, that many important applications can be explored assuming that $A$ is the time to an event.

A different approach that accounts for time-dependent confounding is the structural nested model, which parameterises treatment effects directly in a structural model \cite{robins2014structural}. While this procedure avoids weighting, and will often be more stable and efficient, it relies on other parametric assumptions and can be harder to implement \cite{vansteelandt2016revisiting}.

We conjecture that there is a similar consistency result as Theorem \ref{thm:AalenConsistency} when the outcome model is a weighted Cox regression. However, using a Cox model in the hypothetical scenario after marginalisation leads to restrictions on the data generating mechanisms that are not properly understood, see e.g.\ \cite{havercroft2012simulating}. This issue is related to the non-collapsibility of the Cox model, and it is a problem regardless of the weights being used are continuous or discrete.



\section{Funding}
The authors were all supported by the research grant NFR239956/F20 - Analyzing clinical health registries: Improved software and mathematics of identifiability.

      \section*{Appendix: proofs}
      We need some lemmas to prove Theorem \ref{thm:AalenConsistency}. 
            \begin{lemma} \label{lem:uniform1}
                      Suppose that $\{V^i\}_i$ are processes 
                      on $[0, T]$
                      such that
                              $\sup_ i E\big[  \sup_{s} | V^i_s    | \big]
                                      < \infty $, 
                                      then 
                                      \begin{equation} \label{eq:jhjhjh}
                              \lim_{a \rightarrow \infty} \sup_n  P \bigg(   \sup_s
                      \big| \frac 1 n \sum_{i=1}^n
                      V^{i}_{s} \big|  \geq a  
                      \bigg) = 0.  
                      \end{equation} 

                      \begin{proof}
                              By Markov's inequality, we have for
                              every $a > 0$ that
  \begin{align*}
                                     & P \bigg(   \sup_s \big| \frac 1 n \sum_{i=1}^n
                      V^{i}_{s} \big|  \geq  a 
                      \bigg) 
                      \leq \frac{1}{n a}
                        \sum_{i=1}^n  E_P \bigg[\sup_s \big| V^{i}_{s}
                       \big|\bigg], 
               \end{align*}
               which proves the claim. 
 
                      \end{proof}
              \end{lemma}

\begin{lemma}[A perturbed law of large numbers]\label{lem:perturbedwlln}
               Suppose
         \begin{enumerate}[I)]
                         
                 \item $p^{-1} + q^{-1} = 1$, $p < \infty$, 

                 \item $\{V_i\}_i \subset L^p(P)$,
                         $\{ S_i\}_i \subset L^q(P)$ such that $\{ (V_{{i}},S_{{i}})\}_i$ {is } i.i.d.{, and $V_i$, $S_i$ are measurable with respect to a $\sigma$-algebra $\mathcal{F}_i$},

                 \item Triangular array  $\{S_{{(i,n)} }\}_{n,i \leq n }$ such that 
                         \begin{equation}
                                \lim\limits_{n\longrightarrow \infty} P \big( |S_{{(1,n)}}- S_{{1}}|   \geq
                                 \epsilon  \big) = 0
                         \end{equation}
                         for every $\epsilon > 0$, and there exists a
                         $\tilde S \in L^q(P)$ such that $\tilde S \geq
                         |S_{{(1,n)}}|$ for  every $n$, 
                 \item \label{condition:i independence} The conditional density of $S_{{(i,n)}}$ given   $ {\mathcal{F}_i}$
                         does not depend on $i$.  
                        \end{enumerate}
               This implies that  
               \begin{equation}
                       \lim\limits_{n\longrightarrow \infty}    E\bigg[ \bigg| \frac 1 n \sum_{i = 1} ^n  
                       S_{{(i,n)}} V_{{i}}  -E_{ P} [S_{{1}} V_{{1}}]      \bigg| \bigg] =
                       0. 
               \end{equation}
               
               \begin{proof} 
                From the triangle inequality and condition $\ref{condition:i independence})$ we have that
                       \begin{align*}
                &  E \bigg[\bigg| \frac 1 n \sum_{i = 1}^n   S_{{(i,n)}} V_{{i}} 
                 -   \frac 1 n \sum_{i = 1}^n   S_{{i}} V_{{i}} 
                 \bigg|
                 \bigg] \leq   \frac 1 n \sum_{i = 1}^n 
                   E \big[\big| 
                   \big(  S_{{(i,n)}}  -   S_{{i}}  \big)V_{{i}} 
                    \big|
                 \big] \\
                                     = & 
                   E \big[\big| \big(  S_{{(1,n)}}  -   S_{{1}}  \big)V_{{1}} \big|\big]. 
         \end{align*}
                         The dominated convergence theorem implies that 
                         the last term converges to $0$. Finally, the weak law
                         of large numbers and the triangle inequality yields
\begin{align*}
       &  \lim\limits_{n\longrightarrow \infty}   E\bigg[ \bigg| \frac 1 n \sum_{i = 1} ^n  
                       S_{{(i,n)}} V_{{i}}  -E_{ P} [S_{{1}} V_{{1}}]      \bigg| \bigg]
                        \\\leq  & \lim\limits_{n\longrightarrow \infty}  
                       E \bigg[\bigg| \frac 1 n \sum_{i = 1}^n   S_{{(i,n)}} V_{{i}} 
                 -   \frac 1 n \sum_{i = 1}^n   S_{{i}} V_{{i}}
                 \bigg|
                 \bigg]+ 
                 E \bigg[\bigg| \frac 1 n \sum_{i = 1}^n   S_{{i}} V_{{i}}
                 -      E[S_{{1}} V_{{1}}] 
                 \bigg|
                 \bigg]
                     = 0. 
\end{align*}
                         
               \end{proof}
\end{lemma}

       \begin{lemma} \label{lem:concentrationLem}
               $\{V_i\}_i$ i.i.d. non-negative variables in $L^2(P)$, then 
               \begin{equation}
                       \lim_{n \rightarrow \infty} 
               P \bigg( \frac 1  n \max_{ i \leq n} V_i 
               \geq \epsilon   \bigg) = 0
               \end{equation}
               for every $\epsilon > 0$. 
               \begin{proof}
                      Note that 
\begin{align*}
          P \bigg( \frac 1  n \max_{ i \leq n} V_i 
               > \epsilon   \bigg) &= 
               1- P \bigg( \max_{ i \leq n} V_i
               \leq  \epsilon n   \bigg) 
               =  1- P \bigg( V_1\leq  \epsilon n   \bigg)^n  \\
               &=  
                1- \bigg(1 - P \big( V_1 >  \epsilon n
                \big)\bigg)^n 
\end{align*}
 If  $n > \| V_1  \|_2\epsilon^{-1}$, we therefore have by Chebyshev's
 inequality that 
 \begin{align*}
       P \bigg( \frac 1  n \max_{ i \leq n} V_i 
               > \epsilon   \bigg)    
                \leq 1- \bigg(1 - 
                \frac{E[V_1^2]} {n^2 \epsilon^2} 
                \bigg)^n,
 \end{align*}
 where the last term converges to $0$ when $n \rightarrow \infty$
 since $\lim\limits_{n\longrightarrow \infty} n \log\big( 1 -    \frac{E[V_1^2]} {n^2 \epsilon^2} 
 \big) = 0$ for every $\epsilon > 0$.
               \end{proof}

       \end{lemma}

\begin{lemma} \label{lem:intConv}
        Define $\gamma_s^i := Y_s^{i,D} X^i_s \cdot b_s$, where $X^i_{s}$
                        is  the $i$'th row of $X^{(n)}_{s}$. 
                        If 
the assumptions of Theorem \ref{thm:AalenConsistency} are satisfied, then 
                                          \begin{equation}    \label{eq:integrand}                                           \lim\limits_{n\longrightarrow \infty} P \bigg( \sup_t \bigg|   
                                               \int_0^t  \Gamma^{(n)-1} 
                                               \frac 1 n \sum_{i=1}^n
       R^{(i,n)}_{s-}X_{s-}^{i \intercal}  ( \lambda^{i,D}_s -  \gamma^i_s
       )
    ds \bigg| \geq \delta \bigg)= 0
    \end{equation}
                                      for every $\delta > 0$.

        \begin{proof}
               Assumption $\ref{eq:boundedmatrix})$ from Theorem \ref{thm:AalenConsistency} and Lemma \ref{lem:uniform1}
                implies that
                \begin{equation} \label{eq:bo}\lim_{ J \rightarrow \infty} \inf_n
                                      P\bigg( \sup_t  
                                      \big| \Gamma_{t}^{(n)-1}
                                               \frac 1 n \sum_{i=1}^n
       R^{(i,n)}_{t-}X_{t-}^{i \intercal}  ( \lambda^{i,D}_t -  \gamma^i_t
       )\big|  > J \bigg) =  0.  
                                      \end{equation}
Moreover, Lemma \ref{lem:perturbedwlln} implies that  
                \begin{align*}
                          \frac 1 n \sum_{i=1}^n
       R^{(i,n)}_{t-}X_{t-}^{i \intercal}  ( \lambda^{i,D}_t -  \gamma^i_t
       )
                \end{align*}
        converges in probability to 
        \begin{align*}
         E_P  \big[R^1_{t-} X_{t-} ^{1 \intercal} \cdot ( \lambda^{1,D}_t - \gamma_t^1 )
\big]
\end{align*}
However, from the innovation theorem we have that this equals
\begin{align*} 
E_{\tilde P}  \big[X_{t-} ^{1 \intercal} \cdot ( \lambda^{1,D}_t - \gamma_t^1 )\big] 
= E_{\tilde P}  \big[X_{t-} ^{1 \intercal} \cdot (E_{ \tilde P} [ \lambda^{1,D}_t 
|\mathcal F_{t-}^{1 ,\mathcal V_0} ]
- \gamma_t^1 )\big] = 0, 
\end{align*}
since $X_{t-}^1$ and $\gamma_t^1$ are $\mathcal{F}_{t-}^{1,\mathcal{V}_0}$ measurable. This and \eqref{eq:bo} enables us to apply  \cite[Lemma II.5.3]{Andersen} to
obtain \eqref{eq:integrand}.

        \end{proof}
\end{lemma}

\begin{lemma}\label{lem:martConv}
       Suppose that  $\ref{eq:lambdabound}) $ and $\ref{eq:boundedmatrix})$ from Theorem \ref{thm:AalenConsistency}
are satisfied and 
let $M^{(n)}_t := \begin{pmatrix}N_t^{1,D} - \int_0^t \lambda_s^{1,D} ds , \dots, N_t^{n,D}
- \int_0^t \lambda_s^{n,D} ds \end{pmatrix}^\intercal$. Then  
\begin{equation} \label{eq:varConvergence}
        \Xi^{(n)}_t :=  \frac 1 n \int_0^t \Gamma^{(n)-1}_s
       X_{s-}^{(n)\intercal} Y^{(n),D}_{s} d
       M_s^{(n)} 
\end{equation}
defines a square integrable local martingale with respect to the filtration 
$\mathcal F_{s}^{1,\mathcal V_0 \cup \mathcal L} \otimes \dots \otimes \mathcal
F_{s}^{n,\mathcal V_0 \cup \mathcal L} $ and
\begin{equation}
        \lim_{n \rightarrow \infty}  
P\bigg(   \Tr ( \langle   \Xi^{(n)}\rangle_T)  \geq \delta \bigg) = 0
\end{equation}
        for every $\delta >0$.

        \begin{proof}
               Writing $\lambda^{(n)}$ for the diagonal matrix with $i$'th diagonal element equal to $\lambda^{i,D}$, we have that
                   \begin{align} \label{eq:Lenglart}
                           \Tr (\langle \Xi^{(n)} \rangle_T)  =   & 
   \int_0^T \frac {1} {n^2}  \Tr \bigg(\Gamma^{(n)-1}_s
      X_{s-}^{(n)\intercal}  Y^{(n),D}_{s} \lambda_s^{(n)}   Y^{(n),D}_{s} 
        X_{s-}^{(n)}  \Gamma^{(n)-1}_s\bigg) 
       ds. 
       \end{align}
      Moreover, 
      \begin{align}
              &  \frac {1} {n^2}  \Tr \bigg(\Gamma^{(n)-1}_s
      X_{s-}^{(n)\intercal}   Y^{(n),D}_{s} \lambda_s^{(n)}   Y^{(n),D}_{s} 
        X_{s-}^{(n)}  \Gamma^{(n)-1}_s\bigg) \\ \leq  &
        \frac {1} {n^2}  \Tr \bigg(\Gamma^{(n)-1}_s
      X_{s-}^{(n)\intercal}   Y^{(n),D}_{s} 
        X_{s-}^{(n)}  \Gamma^{(n)-1}_s\bigg) \cdot \max_{i \leq n}   Y_s^{i,D}
        R^{(i,n)}_{s-} 
        \lambda_s^{i,D}   \\ \leq  & 
         \Tr \bigg(\Gamma^{(n)-1}_s
     \bigg) \cdot \big(\frac 1 n  \max_{i \leq n}   
        \lambda_s^{i,D} \big) \cdot \|R^{(i,n)}\|_\infty \label{eq:towards0}
        \\ \leq  & 
         \Tr \bigg(\Gamma^{(n)-1}_s
     \bigg) \cdot \big(\frac 1 n  \sum_{i \leq n}   
        \lambda_s^{i,D} \big) \cdot \|R^{(i,n)}\|_\infty \label{eq:LLNbound}
\end{align}
      
      Now, $\ref{eq:boundedmatrix})$, \eqref{eq:LLNbound} and Lemma \ref{lem:uniform1} implies that  
\begin{align*}
        \lim_{J \rightarrow \infty} \inf_n P \bigg( \sup_s  \frac {1} {n^2}  \Tr \bigg(\Gamma^{(n)-1}_s
      X_{s-}^{(n)\intercal}  Y^{(n),D}_{s} \lambda_s^{(n)}  Y^{(n),D}_{s} 
        X_{s-}^{(n)}  \Gamma^{(n)-1}_s  \bigg) \geq J \bigg)  = 0.
\end{align*}
On the other hand, Lemma
\ref{lem:concentrationLem}, \eqref{eq:towards0} and $\ref{eq:boundedmatrix})$  gives us that 
\begin{align*}
        \lim_{n\rightarrow \infty}   P \bigg(  \frac {1} {n^2}  \Tr \bigg(\Gamma^{(n)-1}_s
      X_{s-}^{(n)\intercal}  Y^{(n),D}_{s} \lambda_s^{(n)}  Y^{(n),D}_{s} 
        X_{s-}^{(n)}  \Gamma^{(n)-1}_s \bigg) \geq \delta \bigg)  = 0
\end{align*}
for every $s$ and $\delta >0$, so \cite[Propositon II.5.3]{Andersen} implies that \eqref{eq:varConvergence} also holds.
        \end{proof}
\end{lemma}

              \begin{proof}[Proof of Theorem \ref{thm:AalenConsistency}]
We have the following decomposition:  
\begin{align*}
         B^{(n)}_t   - B_t = &\int_0^t  (X_{s-}^{(n)\intercal }
        Y^{(n),D}_{s} X_{s-}^{(n)})^{-1}        \big(
        X_{s-}^{(n)\intercal} Y^{(n),D}_{s}   \lambda^{(n)}_s -
        X_{s-}^{(n)\intercal} Y^{(n),D}_{s}   X_{s-}^{(n)} b_s \big) ds
                        \\
       & + \int_0^t (X_{s-}^{(n)\intercal }  Y^{(n),D}_{s}
       X_{s-}^{(n)})^{-1} X_{s-}^{(n)\intercal}  Y^{(n),D}_{s} d
       M_s^{(n)} \\ = & \label{eq:kkjh} 
       \int_0^t  \Gamma^{(n)-1}      \frac 1 n \sum_{i=1}^n
        R^{(i,n)}_{s-}X_{s-}^{i \intercal}  ( \lambda^{i,D}_s -  \gamma^i_s
       )
    ds  + \Xi_t^{(n)}.
\end{align*}
     Lenglarts inequality \cite[Lemma I.3.30]{JacodShiryaev} together with Lemma \ref{lem:martConv} implies that $\Xi^{(n)}$ converges 
      uniformly in probability to $0$. 
     Moreover,  Lemma  \ref{lem:intConv}
  implies that 
    $ \int_0^\cdot  \Gamma^{(n)-1}      \frac 1 n \sum_{i=1}^n
        R^{(i,n)}_{s-}X_{s-}^{i \intercal}  ( \lambda^{i,D}_s -  \gamma^i_s
       )
    ds$
  converges in same sense to $0$, which proves the consistency.

To see that $B^{(n)}$ is P-UT, note that it coincides 
with the sum of $B_t$, 
$\Xi^{(n)}$ and $\int_0^\cdot  \Gamma_s^{(n)-1}      \frac 1 n \sum_{i=1}^n
        R^{(i,n)}_{s-}X_{s-}^{i \intercal}  ( \lambda^{i,D}_s -  \gamma^i_s
       )
    ds$. According to \cite[Lemma 1]{ryalen2018transforming}, the latter is P-UT since $\ref{eq:boundedmatrix})$ and Lemma \ref{lem:uniform1} implies  \eqref{eq:tightRho}. 
    Moreover, $B_t = \int_0^\cdot b_s ds$ is clearly
P-UT, since $b_t$ is uniformly bounded.
$\Xi^{(n)}$ is also P-UT since Lemma \ref{lem:martConv} implies that \eqref{eq:PUTMart} is
satisfied. Finally, as $B^{(n)}$ is a sum of three processes that are P-UT, it is
necessarily P-UT itself.


      \end{proof}

      \section*{Proof of Theorem \ref{thm:ahwConsist}}

\begin{lemma} \label{lemma:K(ni)}
 Suppose that \ref{assumpt:supL2}. and  \ref{assumpt:noColinearity}. from Theorem \ref{thm:ahwConsist} are satisfied, and that
\begin{enumerate}[I)]
        \item  \label{item:tightthetan}
                \begin{equation*} 
                        \lim_{a \rightarrow \infty}
                \sup_n P \bigg(    \sup_t \big| \theta^{(i,n)}_t  \big| \geq a
                \bigg) = 0, 
        \end{equation*}
 
        \item  \label{item:thetaconvergence} $\theta_{t-}^{(i,n)}$
                converges to  $\theta^i_t$ in probability for each $i$ and $t$. 
     \end{enumerate}

Then we have that $K^{(i,n)}$ is predictably uniformly tight (P-UT) and 
\begin{equation}
        \lim_{n} P\bigg(  \sup_t \big| K_t^{(i,n)} - K_t^i   \big|   \geq \epsilon \bigg) = 0
\end{equation}
for every $i$ and $\epsilon > 0$.  

\begin{proof}
Note that 
\begin{equation}
        K_t^{(i,n)} - K_t^i = 
         \int_0^t ( \theta^{(i,n)}_{s-}  -  \theta_s^i ) dN_s^{i,A} 
        + n^{-1/2}  \int_0^t Y_s ^i Z_{s-}^{i\intercal} d W^{(n)}_s
        - n^{-1/2} \int_0^t  Y_s^i \tilde Z_{s-}^{i\intercal} d
        \tilde W^{(n)}_s,
\end{equation}
where 
$W_t^{(n)} := n^{1/2}(H^{(n)}_t - H_t)$ 
and 
$\tilde W_t^{(n)} :=n^{1/2}( \tilde H^{(n)}_t - \tilde H_t)$ 
are square-integrable martingales with respect  {to}
 $\mathcal F_{t}^{1,\mathcal  V_0 \cup \mathcal L  } \otimes \cdots \otimes
 \mathcal F_{t}^{n,\mathcal V_0 \cup \mathcal L  }$ and 
  $\mathcal F_{t}^{1,\mathcal  V_0  } \otimes \cdots \otimes \mathcal F_{t}^{n,\mathcal V_0 
          }$
respectively.

               Let $\tau$ be an optional stopping time and note that 
        \begin{align*}
                E\bigg[ \bigg| \int_0^\tau (\theta^{(i,n)}_{s-}  - \theta^i_s)
                dN_s^{i,A} \bigg|
                \bigg] \leq E\bigg[  \int_0^\tau \big|\theta^{(i,n)}_{s-}  - 
                \theta^i_s 
                \big| dN_s^{i,A}
                \bigg] = E\bigg[  \int_0^\tau \big|\theta^{(i,n)}_{s-}  - \theta^i_s 
                \big| \lambda_s^{i,A} ds 
                \bigg], 
        \end{align*}
        so by Lenglarts inequality, \cite[I.3.30]{JacodShiryaev}, we see that 
        \begin{equation} \label{eq:consistNNN}
                \lim\limits_{n\longrightarrow \infty} P \bigg( \sup_{t \leq T} \bigg| \int_0^t (\theta^{(i,n)}_{s-}  - \theta^i_s)
                dN_s^{i,A} \bigg| \geq \epsilon \bigg) = 0
        \end{equation}
        for every $\epsilon > 0$ if
         \begin{equation}
                \lim\limits_{n\longrightarrow \infty} P \bigg(  \int_0^T \big| \theta^{(i,n)}_{s-}  -
                \theta^i_s \big|
                \lambda_s^{i,A} ds \geq \epsilon \bigg) = 0,
        \end{equation}
        for every $\epsilon > 0$. The latter property holds due to $\ref{item:tightthetan})$, $\ref{item:thetaconvergence})$ and
        \cite[Proposition II.5.3]{Andersen}.

Since $\{ \int_0^t Y_s^i Z_{s-}^{i\intercal}   dW^{(n)}_s  \}_n$ converges in the skorokhod topology, we have that $\{ \sup_{t \leq T}|\int_0^t Y_s^i Z_{s-}^{i\intercal}   dW^{(n)}_s |  \}_n$ is tight \cite[Theorem VI.3.21]{JacodShiryaev}. Therefore, we also get that
\begin{equation} \label{eq:consistM}
        \lim\limits_{n\longrightarrow \infty} P\bigg( \sup_{t \leq T}  |n ^{-1/2} 
        \int_0^t Y_s^i Z_{s-}^{i\intercal}   dW^{(n)}_s |\geq \epsilon \bigg) = 0
\end{equation}
for every $\epsilon > 0$. For the same reason we also have
\begin{equation} \label{eq:consisttildeM}
        \lim\limits_{n\longrightarrow \infty} P\bigg( \sup_{t \leq T}  | n ^{-1/2} \int_0^t Y_s^i \tilde 
        Z_{s-}^{i \intercal}   d\tilde W^{(n)}_s |\geq \epsilon \bigg) = 0.
\end{equation}
By combining \eqref{eq:consistM},\eqref{eq:consisttildeM} and \eqref{eq:consistNNN}, 
we obtain that  
\begin{equation}
        \lim\limits_{n\longrightarrow \infty} P\bigg( \sup_{t \leq T} | K_t ^{(i,n)} - K_t^i |\geq \epsilon \bigg) = 0 
\end{equation}
for every $\epsilon > 0$.

To see that $K^{(i,n)}$ is P-UT, note that 
the compensator of $\int_0^\cdot (\theta^{(i,n)}_{s-} - 1) dN_s^{i,A}$ equals 
$\int_0^\cdot (\theta^{(i,n)}_{s-} - 1) \lambda_s^{i,A} ds$ and 
\begin{equation*}
        \langle \int_0^\cdot (\theta^{(i,n)}_{s-} - 1) dN_s^{i,A}  - 
        \int_0^\cdot (\theta^{(i,n)}_{s-} - 1) \lambda_s^{i,A} ds \rangle_T = 
\int_0^T (\theta^{(i,n)}_{s-} - 1)^2 \lambda_s^{i,A} ds.
\end{equation*}
The assumptions $\ref{item:tightthetan})$ {in this Lemma} and $\ref{assumpt:supL2})$ in Theorem \ref{thm:ahwConsist}, together with
\cite[Lemma 1]{ryalen2018transforming}  therefore imply that 
$\int_0^\cdot (\theta^{(i,n)}_{s-} - 1) dN_s^{i,A}$ is P-UT.

To see that 
 $
\int_0^\cdot  Y_s ^i \tilde Z_{s-}^{i\intercal} d \tilde H^{(n)}_s $ 
is P-UT, note that 
\begin{equation}\label{eq:An integral is P-UT}
        \int_0^\cdot  Y_s ^i \tilde Z_{s-}^{i\intercal} d \tilde H^{(n)}_s = 
        n^{-1/2}  \int_0^\cdot  Y_s ^i \tilde Z_{s-}^{i\intercal} d \tilde W^{(n)}_s
        +  \int_0^\cdot  Y_s ^i \tilde Z_{s-}^{i\intercal} d \tilde H_s.   
\end{equation}

An analogous decompositon yields that$
\int_0^\cdot  Y_s ^i Z_{s-}^{i\intercal} d  H^{(n)}_s $ 
is P-UT. This means that $K^{(i,n)}$ is a sum of three processes that are
P-UT, and must therefore be P-UT itself. 

  \end{proof}  
\end{lemma}

\begin{lemma} \label{lem:sqrnderivative}
        Suppose that 
\begin{enumerate}[I)]
        \item $\{\kappa_n\}_n$ increasing sequence of positive numbers such that 
                $$\lim\limits_{n\longrightarrow \infty} \kappa_n = \infty \text{ and } 
                \sup_n \frac {\kappa_n }{\sqrt n} < \infty,$$

        \item $h_t$ 
        {is a}
        continuous vector valued function, 
        \item $Z^i $ is caglad with $E[\sup_{t\leq T} |Z_t^i|^3_3  ] < \infty $, 

        \item 
         \begin{equation}
                        \lim_{J \rightarrow \infty} \sup_n P \bigg( \Tr\Big(
                 \big( \frac 1 n   Z^{(n)\intercal}_{t-} 
                 Y^{(n),A}_t Z^{(n)}_{t-})^{-1}    
                 \Big) \geq J \bigg) =  0 \end{equation}
        \item  $Y^{i,A} Z_{\cdot -}^{i\intercal} h$ defines the intensity for
                $N^{i,A}$               with respect to
                $P$ and $\mathcal F_\cdot^{i,\mathcal{V}_0}$. 

                Now, 
                \begin{equation}
                        \lim\limits_{n\longrightarrow \infty} P \bigg(  
                        \sup_{ 1/\kappa_n \leq t \leq T} \Big|   
                        \kappa_n \int_{t - 1/\kappa_n } ^t Y_s^{i,A}
                        Z_{s-}^{i\intercal} dH_s^{(n)} 
                        - Y_t^{i,A} Z_{t-}^{i \intercal} h_t \Big| \geq \epsilon 
                        \bigg) = 0. 
                \end{equation}
\end{enumerate}

\begin{proof}
      Note that 
        \begin{align}
                & \kappa_n   \int_{t-1/\kappa_n}^t Y_s^i Z_{s-}^{i \intercal} dH^{(n)}_s - Y_t^{i} Z_{t-}^{i \intercal} h_t \\  
                = &\frac{\kappa_n }{ \sqrt n }
                \int_0^t Y_s^i Z_{s-}^{i\intercal}  dW^{(n)}_s -  \frac {\kappa_n }{ \sqrt{n} }
                \int_0^{t - 1/\kappa_n } Y_s^i Z_{s-}^{i\intercal}  dW^{(n)}_s  \\ 
                & + \kappa_n \int_{t-1/\kappa_n }^t     Y_s^i Z_{s-}^{i \intercal}
                h_s ds -  Y_t^i Z_{t-}^{i\intercal}  h_t.
        \end{align}

        The martingale central limit theorem implies that $\{ W^{(n)} \}$ is a sequence
        of martingales that converges in law to 
        a continuous Gaussian processes with 
        independent increments, see \cite{Andersen}. 
        Moreover, \cite[Proposition 1]{ryalen2018transforming} says that 
        $\{W^{(n)}\}_n$ is P-UT.  
Therefore
\cite[Theorem VI 6.22]{JacodShiryaev} implies that 
$
  \int_0 ^\cdot Y_s^{i,A} Z_{s-}^{i\intercal}  d
                W^{(n)}_s 
$
converges in law to a continuous process, so it is C-tight. 
Moreover, from \cite[Proposition VI.3.26]{JacodShiryaev} we have that 
\begin{equation} \label{eq:PUTTrick}
        \lim_{n \longrightarrow \infty} P \bigg( \sup_{1/\kappa_n \leq  t \leq T} 
        \Big| 
                        \int_0 ^t Y_s^{i,A} Z_{s-}^{i\intercal}  d
                W^{(n)}_s -  
                \int_0 ^{t - 1/\kappa_n } Y_s^{i,A} Z_{s-}^{i\intercal}  d
                W^{(n)}_s 
        \Big| \geq \epsilon       \bigg) = 0
\end{equation}
for every $\epsilon> 0$. 
The mean value theorem of elementary calculus implies that 
\begin{equation} \label{eq:meanval}
        \lim\limits_{n\longrightarrow \infty}    \sup_{1/\kappa_n \leq  t \leq T} 
        \Big|  \kappa_n   \int_{ t - 1/\kappa_n }  ^t Y_s^{i,A} Z_{s-}^{i\intercal}  
        h_s ds -  Y_t^{i,A} Z_{t-}^{i\intercal}  
                h_t \Big| = 0 \end{equation}
$P$ a.s. 
Combining \eqref{eq:PUTTrick} and \eqref{eq:meanval} yields the claim. 
\end{proof}
\end{lemma}

\begin{proof}[Proof of Theorem \ref{thm:ahwConsist}]
       
               Combining \eqref{eq:adhoc} and the decomposition in the proof of Lemma 
        \ref{lem:sqrnderivative}, we see that 
        \begin{equation}
                 \lim\limits_{n\longrightarrow \infty} P \bigg( \sup_{  1/ \kappa_n \leq t  \leq T}  \bigg|
                 \kappa_n   \int_{t- 1/\kappa_n} ^t Y_s^{i,A}
                        \tilde Z_{s-}^{i\intercal} d
                        \tilde H^{(n)}_s  /\tilde \lambda_t^{i,A}   - 1 \bigg| \geq
                        \epsilon \bigg) = 0.
        \end{equation}
        Combining \eqref{eq:adhoc} and \ref{item:thetaConditions}. we also have
        \begin{equation}
                 \lim\limits_{n\longrightarrow \infty} P \bigg( \sup_{  1/ \kappa_n \leq t  \leq T}  \bigg|
                 \kappa_n   \int_{t- 1/\kappa_n} ^t Y_s^{i,A}
                         Z_{s-}^{i\intercal} d
                         H^{(n)}_s  / \lambda_t^{i,A}   - 1 \bigg| \geq
                        \epsilon \bigg) = 0.
        \end{equation}
         Whenever $t \geq 1/ \kappa_n $, we have that by the
         continuous mapping theorem that 
        \begin{align*}
                & \lim\limits_{n\longrightarrow \infty} P \bigg( \sup_{1/\kappa_n \leq t \leq T  } \big|\theta^{(i,n)}_t- \theta_t^i 
                \big| \geq \epsilon \bigg) \\ = &  
                 \lim\limits_{n\longrightarrow \infty} P \bigg( \sup_{1/\kappa_n \leq t \leq T  } \big|
                \theta_t^i \cdot \bigg(\frac {    \kappa_n   \int_{t- 1/\kappa_n} ^t Y_s^{i,A}
                        \tilde Z_{s-}^{i\intercal} d
                        \tilde H^{(n)}_s  /\tilde \lambda_t^{i,A}      }   {  
                 \kappa_n   \int_{t- 1/\kappa_n} ^t Y_s^{i,A}
                         Z_{s-}^{i\intercal} d
                         H^{(n)}_s  / \lambda_t^{i,A}   } - 1\bigg)   \big| \geq
                 \epsilon \bigg) \\  =& 0. 
        \end{align*}
        Since $\theta^i$ is right-continuous at $t= 0$, we have that 
\begin{equation}
         \lim\limits_{n\longrightarrow \infty} P \bigg( \sup_{0 \leq t \leq T } \big|\theta^{(i,n)}_t- \theta_t^i 
                \big| \geq \epsilon \bigg) = 0.
\end{equation}
    
Finally, \cite[Corollary VI 3.33]{JacodShiryaev} implies that  $\{( R^{(i,n)}_0,K^{(i,n)})\}_n$ converges to
$(R_0^i, K^i)$ in probability. 
Since $K^{(i,n)}$ is P-UT,
$$
R_t^{(i,n)} =  1  + \int_0^t R_{s-}^{(i,n)} dK_{s}^{(i,n)} 
$$
and 
$$
R_t^{i} =  1  + \int_0^t R_{s-}^{i} dK_{s}^{i} 
$$
                        \cite[Theorem IX 6.9 ]{JacodShiryaev} 
                        implies that $R^{(i,n)}$ converges to 
                        $R^{i}$ in probability.

\end{proof}

\bibliography{references}
\bibliographystyle{unsrt}
\end{document}